\documentclass[pra,reprint,superscriptaddress,nofootinbib,nobalancelastpage]{revtex4-2}
\usepackage{graphicx}  % needed for figures
\usepackage{dcolumn}   % needed for some tables
\usepackage{bm}        % for math
\usepackage{amssymb}   % for math
\usepackage{amsmath}
\usepackage{mathtools}
\usepackage{mathbbol} %Some other math symbols like \mathbb{1}
\usepackage{xcolor} % Colored text
\usepackage{tikz}
\usepackage{amsthm}

\newtheorem*{theorem*}{Theorem}

\usepackage{hyperref}
\hypersetup{
	colorlinks   = true,    % Colours links instead of ugly boxes
	urlcolor     = blue,    % Colour for external hyperlinks
	linkcolor    = blue,    % Colour of internal links
	citecolor    = red      % Colour of citations
}

\begin{document}

\graphicspath{{Graphics/}}

\def\({\left(}
\def\){\right)}
\def\av#1{\left\langle #1\right\rangle}
\def\ad{a^\dagger}
\def\a2d{a^{\dagger 2}}
\def\b2d{b^{\dagger 2}}
\def\d#1{#1^\dagger}
\def\bydef{\stackrel{\wedge}=}
\newcommand{\eq}[1]{Eq.~(\ref{eq:#1})}
\def\Eq#1{Equation~(\ref{eq:#1})}
\def\eqs#1#2{Eqs.~(\ref{eq:#1}) \& (\ref{eq:#2})}
\def\eqlist#1#2{Eqs.~(\ref{eq:#1}-\ref{eq:#2})}
\def\Eqs#1#2{Equations~(\ref{eq:#1}) \& (\ref{eq:#2})}
\def\Eqlist#1#2{Equations~(\ref{eq:#1}-\ref{eq:#2})}
\def\fig#1{Fig.\ref{fig:#1}}
\def\figs#1#2{Figs.\ref{fig:#1} \& \ref{fig:#2}}
\def\Fig#1{Figure~\ref{fig:#1}}
\def\Figs#1#2{Figures~\ref{fig:#1} \& \ref{fig:#2}}
\newcommand\tab[1]{Table~\ref{tab:#1}}
\newcommand\sect[1]{Section~\ref{sec:#1}}
\newcommand\bra[1]{\left\langle\,#1\,\right|} %\mid doesn't work with \left, \right
\newcommand\ket[1]{\left|\,#1\,\right\rangle}
\newcommand\scalprod[2]{\left\langle\,#1\,\right|\left.#2\,\right\rangle}
\newcommand\om[1]{\omega_{\scriptscriptstyle #1}}
\newcommand\tr[1]{{\rm Tr}[\mathbf #1]}

\title{Are all Gaussian states also cluster states? \\
	Essential diagnostic tools for continuous-variable one-way quantum computing}

\author{Carlos Gonz\'alez-Arciniegas}
\email{cag2ze@virginia.edu}
\affiliation{Department of Physics, University of Virginia, 382 McCormick Rd, Charlottesville, VA 22904-4714, USA}
\author{Paulo Nussenzveig}
\author{Marcelo Martinelli}
\affiliation{Instituto de F\'{\i}sica, Universidade de S\~ao Paulo, 05315-970 S\~ao Paulo, SP-Brazil} 
\author{Olivier Pfister}
\affiliation{Department of Physics, University of Virginia, 382 McCormick Rd, Charlottesville, VA 22904-4714, USA}

\begin{abstract} % 500 words
Continuous-variable (CV) cluster states are a universal quantum computing platform that has experimentally out-scaled qubit platforms by orders of magnitude. Room-temperature implementation of CV cluster states has been achieved with quantum optics by using multimode squeezed Gaussian states. It has also been proven that fault tolerance thresholds for CV quantum computing can be reached at realistic squeezing levels. In this paper, we show that standard approaches to design and characterize CV  cluster states can miss entanglement present in the system. Such hidden entanglement may be used to increase the power of a quantum computer but it can also, if undetected, hinder the successful implementation of a quantum algorithm. By a detailed analysis of the structure of Gaussian states, we derive an algorithm that reveals hidden entanglement in an arbitrary Gaussian state and optimizes its use for one-way quantum computing.

 \end{abstract}

\maketitle

\section{Introduction}

Continuous-variable (CV) quantum information \cite{Braunstein2005a,Weedbrook2012,Pfister2019} has achieved groundbreaking scalability performance~\cite{Pysher2011,Chen2014,Yokoyama2013,Yoshikawa2016,Asavanant2019,Larsen2019} in the universal, measurement-based, one-way quantum computing (QC) model~\cite{Menicucci2006}. The CVQC model uses Gaussian (in terms of their Wigner function) cluster states as mathematical substrates~\cite{Zhang2006} together with required non-Gaussian resources (such as photon counting measurement or a cubic phase gate) to constitute a universal quantum computation platform~\cite{Lloyd1999,Gottesman2001,Bartlett2002,Menicucci2006,Mari2012}. The idealized CVQC model employs, in lieu of qubits,  spectrally dense qumodes such as the respective eigenstates $\{\ket{s}_q\}_{s\in\mathbb R}$ and $\{\ket{s}_p\}_{s\in\mathbb R}$  of the amplitude-quadrature operator $Q=(a+\ad)/\sqrt2$ and phase-quadrature operator $P=i(\ad-a)/\sqrt2$ of the quantized electromagnetic field, $a$ being the photon annihilation operator.  These quadrature eigenstates are infinitely squeezed states, which require infinite energy and are therefore unphysical. Realistic CVQC employs qumodes in squeezed Gaussian states, generated by SU(1,1) quadratic Hamiltonians. Such states are arbitrarily good approximations to quadrature eigenstates and also allow a fault tolerance threshold for amounts of squeezing of 10-20 dB~\cite{Menicucci2014ft,Fukui2018,Walshe2019} that are experimentally reachable~\cite{Vahlbruch2016}. 

Such reasonably high fault-tolerant squeezing levels may lead to the false impression that the analysis of the corresponding states would be essentially equivalent to their analysis in infinite squeezing limit. This is incorrect. We show in this paper that the infinite squeezing limit fails to capture a subtle but crucial property of Gaussian states, which we call hidden entanglement. We show that hidden entanglement can disrupt quantum computing if unaccounted for in a CV cluster state. However, hidden entanglement can always be detected and can sometimes be corrected.

In order to frame the problem in the most general way, we use the graphical calculus formalism developed by Menicucci, Flammia, and van Loock~\cite{Menicucci2011}, whose gist is that any pure multimode Gaussian state can be described by a unique graph whose vertices denote the qumodes and whose complex-weighted edges denote the interactions between the qumodes sharing an edge. In this formalism, the real parts of the edge weights denote controlled-phase interactions, which are the CV analogs of  controlled-Z gates for qubits.\footnote{Recall that controlled-Z gates define the edges of qubit-vertex graph states, a.k.a.\ cluster states~\cite{Briegel2001}, which enable measurement-based  one-way quantum computing~\cite{Raussendorf2001}.} These real edges therefore define the graph of a CV cluster state~\cite{Zhang2006} which is usable for  measurement-based, universal quantum computing~\cite{Menicucci2006,Gu2009}.\footnote{In contrast, the graph defined by the imaginary edges cannot be used for measurement-based quantum computing as that would require measuring non-Hermitian observables.}

It had been assumed to this day that, if Gaussian local unitaries (GLUs,\footnote{GLUs are operations generated by all single-mode phase-space rotations (i.e., optical phase shifts), single-mode squeezes, and single-mode shears (which are combinations of the former two). GLUs are equivalent to single-mode symplectic operations because Sp(2,$\mathbb R$) $\sim$ SU(1,1) $\supset$ U(1).} which cannot change the entanglement of the state) are applied to make the imaginary edge weights of the graph vanishingly small (e.g.\ in the limit of infinite squeezing), then a valid cluster state is obtained.

In this paper, we show that the above procedure is, in fact, not unique: starting from a given Gaussian state, there are different choices of GLUs that all give a vanishingly small imaginary graph but do give real graphs that have dramatically different entanglement, for example, containing or not disconnected subgraphs, i.e., seemingly separable quantum states. If the vanishingly small imaginary graph is ignored, this creates the paradoxical and unacceptable situation of the same state being GLU-equivalent both to separable states and to completely inseparable states.

This paradox is resolved by realizing that the entanglement that seems to disappear from the real graph under a GLU becomes, in fact, ``hidden'' by being transferred to the imaginary graph. We must therefore ensure that the GLU leads to an imaginary graph that has no edge between two different qumodes.

We report two mathematical results: {\em(i)},  we derive a sufficient mathematical criterion to find Gaussian states that have uncorrectable, i.e., ``irreducible,'' hidden entanglement, which means that not all their imaginary edges can be transferred to real ones under any GLU and that these states cannot be expressed as  cluster states; {\em(ii)}, we derive an analytic algorithm to express a Gaussian state into a valid CV cluster state, i.e., with {\em null} imaginary edge weights. This algorithm is applicable to any Gaussian state and succeeds for all valid CV cluster states. The whole situation is summarized in \fig{van}.
%________________________________________________________________________________________________________________________________
\begin{figure}[htb]
\vskip -.25in
\begin{center} 
\includegraphics[width=1\columnwidth]{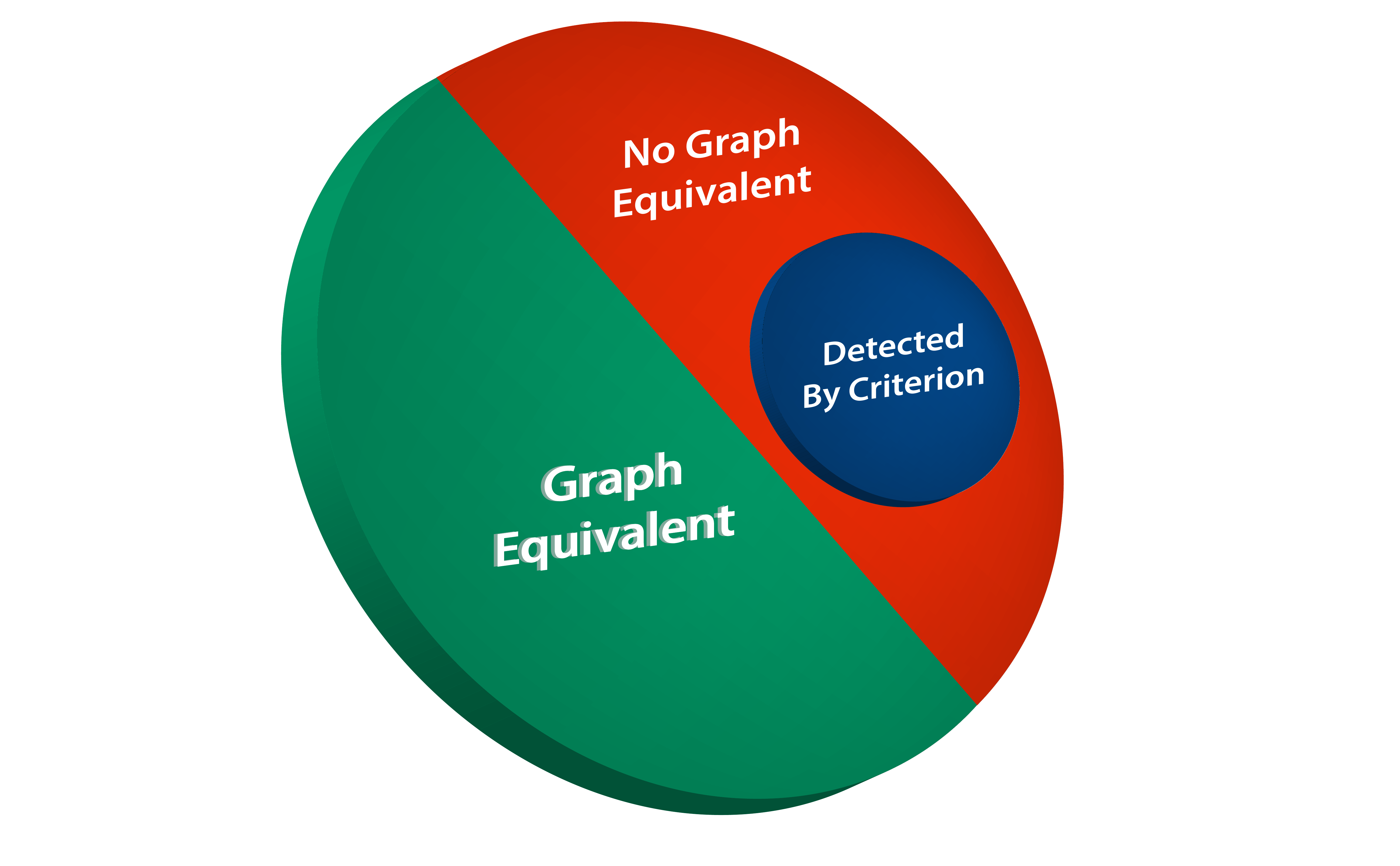}
\end{center}
\vskip -.25in
\caption{\em Venn diagram of all Gaussian states. The two subsets are states which can (green) or can't (red) be cast as a CV cluster state. Some of the invalid graph states are detectable by our sufficient criterion (blue subset). If applied to the whole set, our analytic procedure to express a Gaussian state as a graph state will succeed for the whole green subset.}
\label{fig:van}
\vskip -.25in
\end{figure} 
%________________________________________________________________________________________________________________________________

These results enable us to answer the long-standing question of whether every Gaussian state can be expressed as a cluster state: the answer is negative. 

The paper is structured as follows. In \sect{pb}, we pose the problem of identifying Gaussian states that cannot be expressed as valid cluster states. In \sect{crit}, we derive a sufficient criterion, {\em(i)} above, to identify Gaussian states that are not equivalent to cluster states. In \sect{diag}, we derive a general algorithm, {\em(ii)} above, for expressing a Gaussian state as a valid cluster state. In the final section we summarize our conclusions and discuss their relevance for quantum information science.

\section{The problem: are all Gaussian states valid CV cluster states?}\label{sec:pb}

\subsection{Introduction to the CV cluster state formalism}

\subsubsection{Qubits}

A cluster state~\cite{Briegel2001} $\ket{\!\psi_{V}\!}$ is a graph quantum state~\cite{Hein2004} that contains all the entanglement ever needed for any quantum algorithm~\cite{Raussendorf2001} and that must be sparsely connected in order to be useful for quantum computing~\cite{Bacon2009,Gross2009,Bremner2009}. In the qubit paradigm, a graph state is composed of qubit vertices in the $\ket+=(\ket0+\ket1)/\sqrt2$ state, linked by controlled-Z gate edges. Quantum computation proceeds from a cluster state solely by single-qubit measurements and feedforward to graph neighbors. An $N$-qubit graph state can be defined as a stabilizer state, i.e., a simultaneous eigenstate of the $N$ generators of the stabilizer group with eigenvalue 1. These generators are defined as
\begin{align}
	K_j&=\hat{\sigma}_x^{(j)}\prod_{k=1}^N\left(\hat{\sigma}_z^{(k)}\right)^{{V}_{jk}},\quad j=1,2,...,N,\\
	{K}_j&\ket{\!\psi_{V}\!}=\ket{\!\psi_{V}\!},\qquad j=1,2,...n,
\end{align}
where the ${V}_{jk}$ denote the elements of the adjacency matrix $\mathbf{V}$ of the graph: ${V}_{jk}$=1 if there is an edge between qubits $j$ and $k$ and ${V}_{jk}$=0 otherwise. The operators $\hat{\sigma}^{(j)}_{x,z}$  are Pauli operators acting on qubit $j$. 

\subsubsection{Qumodes}

An ideal CV graph state $\ket{\!\Psi_{V}\!}$ is constructed by preparing each qumode vertex in a zero-phase eigenstate $\ket0_p=\int\ket{s}_q ds/\sqrt{2\pi}$ and by then applying controlled phase-displacements as per the adjacency matrix $\prod_{j\geq k}^N\exp(i\mathbf{V}_{jk}Q_jQ_k)$~\cite{Zhang2006,Menicucci2006}. The generators of the stabilizer group are of the form $K_j=e^{i \alpha \mathcal{N}_j}$,  $\alpha\in \mathbb{R}$, with $\mathcal{N}_j$ denoting the nullifier operators~\cite{Gu2009}
\begin{align}
	\mathcal{N}_j&=P_j-\sum_{k=1}^{N}{V}_{jk}Q_k,\quad j=1,2,...n,\\
	\mathcal{N}_j&\ket{\!\Psi_{V}\!}=0\ket{\!\Psi_{V}\!},\qquad j=1,2,...n,
\end{align}
where $\mathbf V$ is now a weighted adjacency matrix, whose  elements can have any real value. These ideal CV states are infinitely squeezed and therefore unphysical. 

\subsubsection{Finitely squeezed states} 

In the laboratory, one employs the closest approximations to quadrature eigenstates which are Gaussian phase-quadrature-squeezed vacuum states
\begin{align}
\ket{0,r} & 
\propto\int dq\,e^{-\frac12e^{-2r}q^{2}} \,\ket q 
\propto\int dp\,e^{-\frac12e^{2r}p^{2}} \,\ket p
\end{align}
where $r>0$ is the squeezing parameter and coincides with the logarithmic gain of a parametric amplifier such as an optical parametric oscillator (OPO) below threshold. A valid Gaussian approximation of an ideal graph state $|\Psi_{V}\rangle$ of adjacency matrix $\mathbf{V}$ must then fulfill the property 
\begin{align}\label{eq:cov0}
	\text{Cov}\left[\mathbf{P}-\mathbf{V}\mathbf{Q}\right]
	\xrightarrow[r\to\infty]{} \mathbb0
\end{align}
where Cov$[\mathbf{A}]_{jk}=\frac{1}{2}\left\langle\Psi_{V}\right| \{A_j,A_k\}\left| \Psi_{V} \right\rangle$ is an element of the covariance matrix of operator vector  $\mathbf{A}$=($A_1$,...,$A_n$)$^{T}$. 

It was shown by Meniccuci, Flammia, and van Loock~\cite{Menicucci2011} that the effects of finite squeezing can be fully taken into account by defining a complex graph $\mathbf{Z}=\mathbf{V}+i \mathbf{U}$ where $\mathbf{V}$ is the weighted symmetric adjacency matrix as before and $\mathbf{U}$ is a symmetric positive definite matrix that accounts for all finite squeezing effects.\footnote{Note that weighted real graphs also appear in qudit graph state theory but complex graphs are a property of CV systems.} Any pure Gaussian state $\ket{\Psi_{Z}}$\footnote{Without loss of generality, we only consider states with $\av{P_j}=\av{Q_j}=0$ since that doesn't affect the entanglement properties of the system}   can be written in the position representation as  
\begin{align}\label{eq:WaveFunct}
	\scalprod{\mathbf{q}}{\Psi_{Z}}=\Psi_{Z}(\mathbf{q})=\pi^{-N/4}\left(\det\mathbf{U}\right)^{1/4}\exp\left[\frac{i}{2}\mathbf{q}^\text{T}\mathbf{Z}\mathbf{q}\right],
\end{align}
 the positive definiteness of $\mathbf{U}$ ensuring that the state is normalizable. The complex matrix $\mathbf{Z}$ defines exact graph state nullifiers 
\begin{align}
	\left(\mathbf{P}-\mathbf{Z}\mathbf{Q}\right)\ket{\Psi_{Z}}=\mathbf0\ket{\Psi_{Z}}.
\end{align}
These nullifiers are non-Hermitian operators and therefore not suited for measurement-based quantum computation. Nonetheless we can still use the nullifiers given by the adjacency matrix $\mathbf{V}$ if we realize that the wave function of \eq{WaveFunct} satisfies the relation
\begin{align}\label{eq:CovNull}
\text{Cov}\left[\mathbf{P}-\mathbf{V}\mathbf{Q}\right]=\frac{1}{2}\mathbf{U},
\end{align}
and note that, as long as $\mathbf{U}$$\to$$\mathbb0$ or, equivalently, $\tr U$$\to$0 given that $\mathbf{U}$ is positive definite, we recover \eq{cov0}. It has been proposed that having a vanishingly small $\mathbf{U}$ was enough to claim that our Gaussian state approximates an ideal graph state with adjacency matrix $\mathbf{V}$ and hence, we could think of  $\mathbf{U}$ (or $\tr U$) as the error on approximating an ideal CV graph state. Note from \eqs{cov0}{CovNull} that $\mathbf U$=$\mathbb0$ is logically equivalent to the infinite squeezing limit. 
 
If $\tr{U}$ is vanishingly small, we may be tempted to completely disregard $\bf U$ and think of $\bf Z$  as the ideal graph $\mathbf{Z}=\mathbf{V}$. This has been, until now, the main way to deal with Gaussian graph states but, as we show next using the covariance matrix, this definition alone may neglect strong quantum correlations present in the system, which can lead to large errors in the characterization of its entanglement, as we show in \sect{ex}.

\subsubsection{Covariance matrix}

A Gaussian state is fully characterized by its covariance matrix. The quadrature-ordered covariance matrix is
\begin{align}
\mathbf{\Sigma}^{(quad)} = \text{Cov}[\mathbf R] 
&= \begin{pmatrix}
		\text{Cov}[\mathbf Q]&\text{Cov}[\mathbf Q,\mathbf P]\\
		\text{Cov}[\mathbf P,\mathbf Q]&\text{Cov}[\mathbf P]
	\end{pmatrix}\label{eq:cova}\\
&=\frac{1}{2}\begin{pmatrix}
		\mathbf{U}^{-1}&\mathbf{U}^{-1}\mathbf{V}\\
		\mathbf{V}\mathbf{U}^{-1}&\mathbf{U}+\mathbf{V}\mathbf{U}^{-1}\mathbf{V}
	\end{pmatrix},\label{eq:U-1}
\end{align}
where $\bf R$=$(\mathbf Q,\mathbf P)^{T}$. Examination of the upper left block in \eqs{cova}{U-1} yields the formula  
\begin{align}\label{eq:covQ}
	\text{Cov}\left[\mathbf{Q} \right]=\frac{1}{2}\mathbf{U}^{-1} 
\end{align}
and a first observation: {\em the off-diagonal entries of $\mathbf U^{-1}$} may be large, describing strong quantum correlations that could still appear vanishingly small in $\mathbf{U}$ and therefore be missed. We must therefore supplement the $\tr U$$\to$0  criterion [\eq{CovNull}] for finitely squeezed Gaussian states~\cite{Menicucci2011}  by also requiring that $\mathbf U^{-1}$ be diagonal, which is logically equivalent to $\mathbf U$ being diagonal~\cite{Zhu2021}. This will ensure that all the correlations of the system are encoded exclusively in the adjacency matrix $\mathbf{V}$.

Before addressing the diagonalization of $\bf U$ in detail, we first give two concrete examples of hidden entanglement.

\subsection{Examples of hidden entanglement}

\subsubsection{Two qumodes}\label{sec:ex}
We consider here the very simple case of two independent single-mode squeezed states, of respective parameters $r_{1,2}$, interfering at a balanced beamsplitter. Derivation details can be found in the supplemental material~\cite{SupMat}. 
The resulting graph is purely imaginary
\begin{equation}\label{eq:B0f}
\setlength{\unitlength}{.35in}
\begin{picture}(8,1.1)
\thinlines
\put(-.5,-.1){$\ket{\mathcal B(0,r_1,r_2)}_{12}=$}
\put(4,0){\circle*{.5}}
\qbezier(4,0)(2.8,.75)(2.75,0)
\qbezier(4,0)(2.8,-.75)(2.75,0)
\put(7.,0){\circle*{.5}}
\qbezier(7.,0)(8.2,.75)(8.25,0)
\qbezier(7.,0)(8.2,-.75)(8.25,0)
\put(4.25,0){\line(1,0){3}}
\put(4.1,.35){$\begin{matrix}\scriptstyle i(e^{-2r_1}\\ \scriptstyle \quad\ \ - e^{-2r_2})\to0\end{matrix}$}
\put(1.9,.8){$\begin{matrix}\scriptstyle i(e^{-2r_1}\\ \scriptstyle \quad\ \ + e^{-2r_2})\to0\end{matrix}$}
\put(6.6,.8){$\begin{matrix}\scriptstyle i(e^{-2r_1}\\ \scriptstyle \quad\ \ + e^{-2r_2})\to0\end{matrix}$}
\end{picture}
\end{equation}
Since all edges are exponentially decreasing with the squeezing, a logical conclusion would be that the state is, for all intents and purposes, equivalent to the product state obtained in the infinite squeezing limit, i.e., 
\begin{equation}
\setlength{\unitlength}{.35in}
\begin{picture}(8,.5)
\thinlines
\put(-.5,.15){$\ket{\mathcal B(0,r_1,r_2)}_{12}\  \underset{ \scriptscriptstyle r_{1,2}\to\infty}{\longrightarrow}$}
\put(4,0.25){\circle*{.5}}
\put(7.,0.25){\circle*{.5}}
\end{picture}\label{eq:crap}
\end{equation}
This conclusion is incorrect. While state $\ket{\mathcal B(0,r,r)}_{12}$ is, indeed, an exact  product state (which was demonstrated experimentally~\cite{Bruckmeier1997}), state $\ket{\mathcal B(0,r_1,r_2\neq r_1)}_{12}$ can, however, be strongly entangled. We now prove this.

An independent quantitative bipartite entanglement criterion is the generalization of the Peres-Horodecki partial transpose criterion~\cite{Peres1996,Horodecki1996} to continuous variables~\cite{Simon2000} (see also Ref.~\citenum{Duan2000}). Bipartite CV  nonseparability is characterized by the symplectic eigenvalues of the covariance matrix ${\mathbf{\tilde\Sigma}}$ of the partially transposed density operator, which is equivalent to a phase space reflection~\cite{Simon2000}: if the original covariance matrix is ${\mathbf{\Sigma}}$, then ${\mathbf{\tilde\Sigma}}=\mathbf{\Lambda}\mathbf{\Sigma}\mathbf{\Lambda},$ where $\mathbf{\Lambda}$=diag$(1,1,-1,1)$. Entanglement is present if at least one of the symplectic eigenvalues of ${\mathbf{\tilde\Sigma}}$ is less than $\frac12$, their product being $\frac14$. These symplectic eigenvalues are defined as the absolute values of the eigenvalues of $i\mathbf{\tilde\Sigma}\mathbf{\Omega}$ with 
$
	\mathbf{\Omega}=\begin{psmallmatrix}
	\mathbb{0}&\mathbb{1}\\
	-\mathbb{1}&\mathbb{0}
	\end{psmallmatrix}. 
$
The symplectic eigenvalues of $\ket{\mathcal B(0,r_1,r_2)}_{12}$ 
are
\begin{align}\label{eq:SimpEigB}
    \lambda_\pm=\tfrac12\,e^{\pm(r_1-r_2)},
\end{align} 
which shows that the state is a product state iff  $r_1$=$r_2$ but can be significantly entangled if the difference $r_1$-$r_2$ is large. Let's take the case $r_1$=$2r_2$ with $r_2$ already large, e.g.\ 20 and 10 dB squeezing. Then the amount of entanglement in $\ket{\mathcal B(0,2.30,1.65)}_{12}$ is equivalent to that present in a 10 dB-squeezed two-mode-squeezed state,even though all edges in \eq{B0f} are vanishing.

Thus, the null graph edges given by $\bf V$ clearly fail to give the proper description here, even though $\tr U$$\to$0 in both cases. The key points are that the infinite squeezing limit incorrectly symmetrizes the situation, giving the wrong description for finite squeezing, and that the vanishingly small, {\em yet nonzero} off-diagonal entries of $\bf U$ still yield large $\av{Q_{i}Q_{j}}$ correlations from $\mathbf U^{-1}$, as per \eq{U-1}, which are key to the discrepancy.

If we use single-mode symplectic transformations, i.e.\ GLUs, to diagonalize $\bf U$, the entanglement can't be changed~\cite{Simon2000,Adesso2006}. We therefore seek GLUs that transform $\ket{\mathcal B(0,r_1,r_2\neq r_1)}_{12}$ into a two-mode-squeezed state of equal entanglement: the result is~\cite{SupMat}
\begin{equation}\label{eq:booyah}
\ket{\mathcal E\left(r_-\right)}_{12}\!=\!S_1\left(r_+\right)S_2\left(r_+\right)\!\ket{\mathcal B(0,r_1,r_2)}_{12},
\end{equation}
where $S(r)$=$\exp[\frac r{2}(a^{\dag2}-a^2)]$ is a phase squeezing operator (for $r$$>$0) and $r_\pm=\left(r_1\pm r_2\right)/2$. 
\Eq{booyah} solves the conundrum: in finitely squeezed states, entanglement can be hidden by single-mode squeezing in non-symmetric cases where the squeezing is not evenly distributed between the modes. This situation is always absent in the infinitely squeezed case which is, by force, totally symmetric.

After single-mode squeezing operations, \eq{booyah}, and phase shifts carry out $\bf Z$$\mapsto$$\mathbf Z'$, we obtain [see \eq{T1}]
\begin{align}
\tr {U'} = 2\, \text{sech} (r_1-r_2) \underset{r_1>r_2\gg1}{\longrightarrow}4 e^{-(r_1-r_2)},
\end{align}
which can be compared to the value of \eq{B0f}, before the single-mode squeezing operations,
\begin{align}
\tr U 
&= 2(e^{-2r_1}+e^{-2r_2})
\underset{r_1>r_2\gg1}{\longrightarrow}2 e^{-2r_2}.
\end{align}
While both traces tend to zero in the infinite squeezing limit, graph  $\bf V'$ [\eq{T2} with $r$=$r_{-}$] reveals the entanglement of the state whereas graph $\bf V$ [\eqs{B0f}{crap}] does not.

Another important instance of quantum state distortion invisible in the infinite squeezing limit is the generation of two-mode-squeezed state from the interference of orthogonal single-mode-squeezed states: 
\begin{align}
\ket{\mathcal B(\tfrac\pi2,r,r)}=\ket{\mathcal E(r)},
\end{align}
which is relevant to a multitude of CVQI experiments, a few prominent examples of which are Refs.~\citenum{Ou1992,Furusawa1998,Yokoyama2013,Yoshikawa2016,Asavanant2019,Larsen2019}. If the states aren't identically squeezed, we get \begin{align}\label{eq:badtms}
\ket{\mathcal B(\tfrac\pi2,r_{1},r_{2})}=S_{1}(-r_{-})S_{2}(-r_{-})\ket{\mathcal E(r_{+})},
\end{align}
which features, on top of the desired two-mode squeezing by $r_{+}$, excess uncorrelated quantum noise on each qumode, each being independently antisqueezed by $r_{-}$. 

\subsubsection{Six qumodes}

We now turn to a highly multipartite example and show that the two dramatically different graphs of \fig{DiscGraph} are, in fact, GLU-equivalent.
%________________________________________________________________________________________________________________________________
\begin{figure}[htbp]
\centerline{\includegraphics[width=\columnwidth]{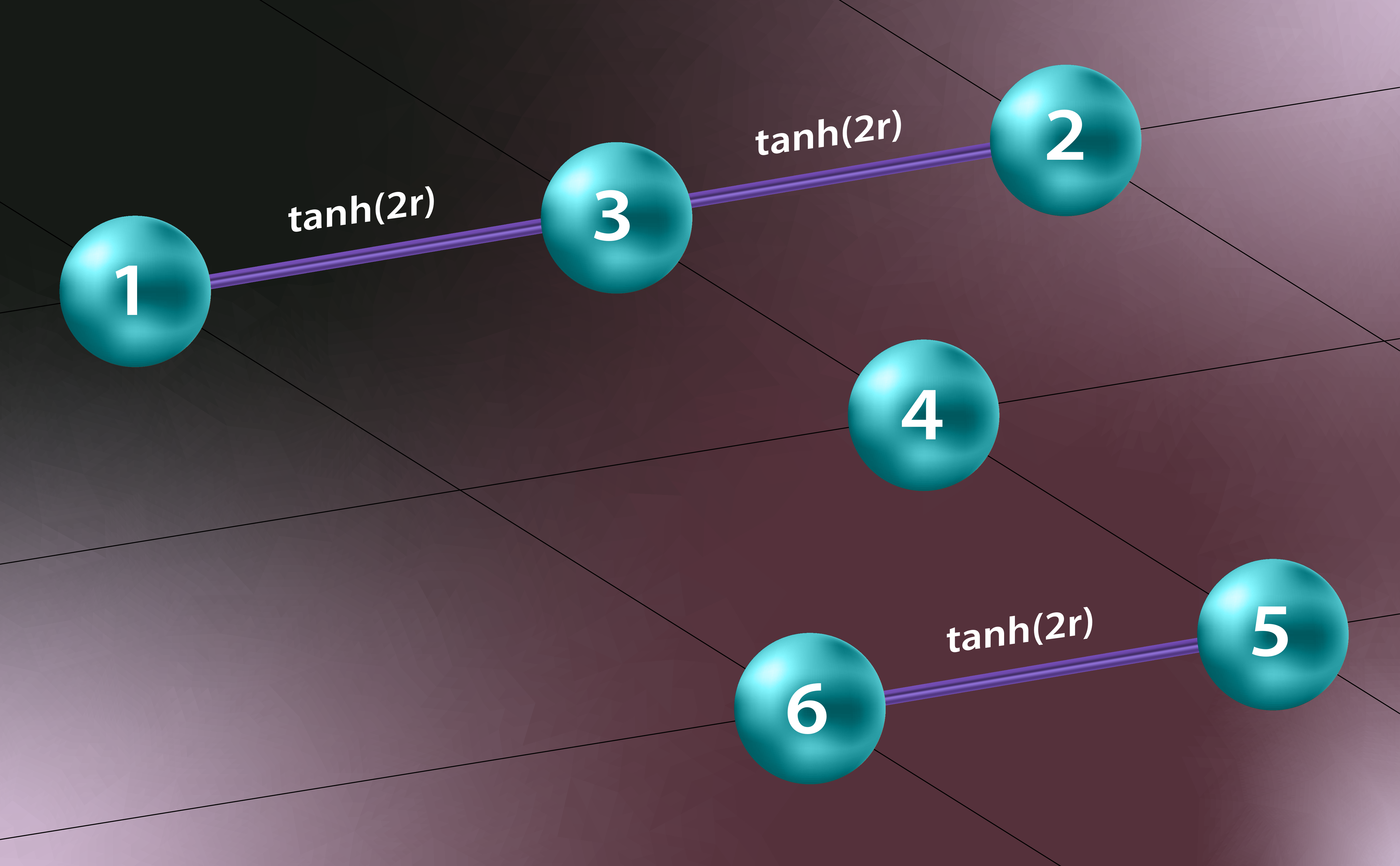}}
\centerline{\includegraphics[width=\columnwidth]{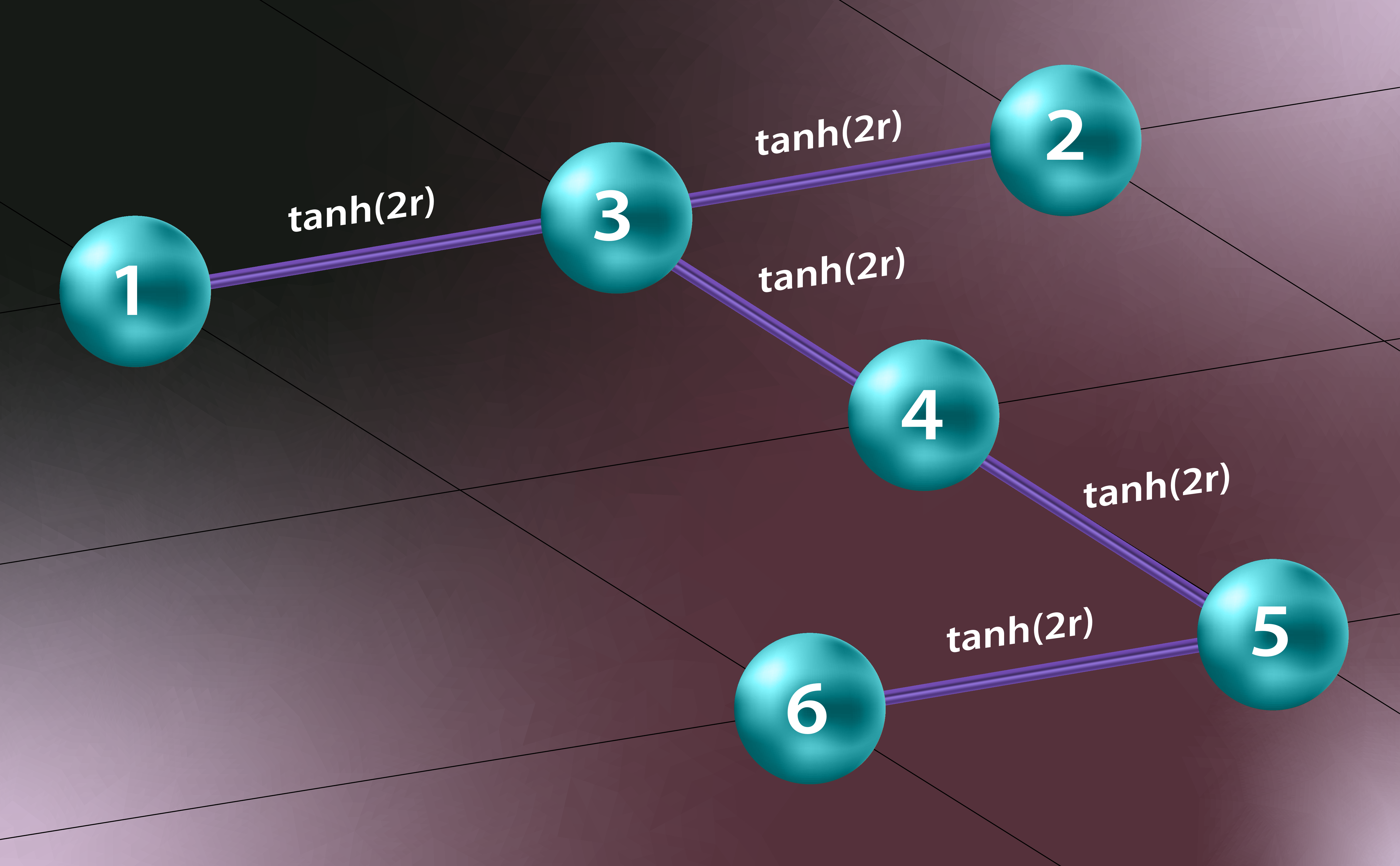}}
\caption{Top, cluster state graph given by matrix $\mathbf{V}$ in \eq{DiscVmatrix}. Note the three disconnected subgraphs. Bottom, cluster state graph given by matrix $\mathbf{V'}$ in \eq{vp}, obtained by diagonalizing $\bf U'$ by GLUs. The state is now one connected graph, all hidden entanglement having been revealed.}
\label{fig:DiscGraph}
\end{figure}
%________________________________________________________________________________________________________________________________

\Fig{DiscGraph}, Top, represents the real part (adjacency matrix {\bf V}) of the 6-mode graph $\mathbf{Z}=\mathbf{V}+i \mathbf{U}$: 
\begin{align}\label{eq:DiscVmatrix}
	\mathbf{V}&=
	\begin{pmatrix}
		0 & 0 & t & 0 & 0 & 0 \\
		0 & 0 & t & 0 & 0 & 0 \\
		t & t & 0 & 0 & 0 & 0 \\
		0 & 0 & 0 & 0 & 0 & 0 \\
		0 & 0 & 0 & 0 & 0 & t \\
		0 & 0 & 0 & 0 & t & 0 
	\end{pmatrix}\\
	\mathbf{U}&=
	\begin{pmatrix}
		c^{-1}  & 0 & 0 & 0 & 0 & 0 \\
		0 & c^{-1}  & 0 & 0 & 0 & 0 \\
		0 & 0 & c^{-1}  & -c^{-1}t  & c^{-1}t^2  & 0 \\
		0 & 0 & -c^{-1}t  & c^{-1}  & -c^{-1}t  & 0 \\
		0 & 0 & c^{-1}t^2  & -c^{-1}t  & c^{-1}  & 0 \\
		0 & 0 & 0 & 0 & 0 & c^{-1}  
	\end{pmatrix}.
\end{align}
where $t$=$\tanh 2r$ and $c^{-1}$=$\text{sech}2r$. This state clearly verifies $\mathbf{U}\to \mathbf{0}$ for $r\to \infty$, because $t\to1$ and $c^{-1}\to0$. \Fig{DiscGraph}, Top, shows three disconnected graphs, the chain 1-3-2, the chain 5-6 and sole qumode 4.  \Fig{DiscGraph}, Top, however, is incorrect: there are hidden quantum correlations, or hidden entanglement, between modes 3, 4, and 5, which are present in $\bf U$ and will disrupt quantum computation over the cluster state. For example, a measurement of isolated qumode 4 isn't expected to affect the rest of the graph since it's disconnected from it. Yet, it will. These hidden correlations can actually be quite strong, even with $\tr U\to0$. We can see this from \eq{covQ}. In this instance we have
\begin{align}
	\mathbf{U}^{-1}=
	\begin{pmatrix}
		c & 0 & 0 & 0 & 0 & 0 \\
		0 & c & 0 & 0 & 0 & 0 \\
		0 & 0 & c^3 & s \,c^2 & 0 & 0 \\
		0 & 0 & s \,c^2 & c (c^2+s^2) & s \,c^2 & 0 \\
		0 & 0 & 0 & s \,c^2 & c^3 & 0 \\
		0 & 0 & 0 & 0 & 0 & c \\
	\end{pmatrix}
\end{align}
where $c=\cosh2r$ and $s=\sinh2r$, which are very large. The only way to suppress these spurious correlations is, of course, to diagonalize $\bf U$. This must be done without changing the entanglement of the graph state, i.e., using only local symplectic, equivalently local linear unitary, operations, i.e., GLUs. In this work, we identify all situations where this is feasible and demonstrate a GLU-diagonalization procedure that will succeed whenever possible. In that sense, given any Gaussian state, we are now able to ascertain whether it is a valid cluster state, i.e., with $\bf U$ diagonal and $\tr U\to0$, or not.

Let's go back to our 6-mode example. A Gaussian unitary $\mathcal U$ can be represented in the Heisenberg picture by a symplectic matrix $\bf S$ such that
\begin{align}
	\mathbf{R}^\prime=\mathcal{U}^\dagger\mathbf{R}\,\mathcal{U}=\mathbf{S}\mathbf{R}.
\end{align}
The symplectic symmetry requires that $\mathbf{S}$ fulfill $\mathbf{S}\mathbf{\Omega}\mathbf{S}^{\text{T}}$=$\mathbf{\Omega}$, which is equivalent to preserving the canonical commutation relations $[R_j,R_k]=[R_j^\prime,R_k^\prime]={\Omega}_{jk}$. 
Any GLU must then have the form
\begin{align}\label{eq:glu}
\mathcal U &= \bigotimes_{j=1}^{N}\mathcal U_{j}\\
\mathbf S &= \bigoplus_{j=1}^{N}\mathbf S_{j} \label{eq:glu}
\end{align}
so that
\begin{align}
	\begin{pmatrix}Q_j^\prime\\P_j^\prime\end{pmatrix}=\mathcal{U}_{j}^\dagger\begin{pmatrix}Q_j\\P_j\end{pmatrix}\mathcal{U}_{j}=\mathbf{S}_j\begin{pmatrix}Q_j\\P_j\end{pmatrix}\; j=1,...,N.
\end{align}
In our example, taking $\mathbf{U}$ to a diagonal from can be done with single-mode squeezing GLUs:
\begin{align}
	\mathbf{S}_1=\mathbf{S}_2=\mathbf{S}_3^{-1}=\mathbf{F}\mathbf{S}_4^{-1}
	=\mathbf{S}_5^{-1}=\mathbf{S}_6=
	\begin{psmallmatrix}
		c^{-1}  & 0 \\
		0 & c 
	\end{psmallmatrix},
\end{align} 
where $\mathbf{F}=\begin{psmallmatrix}	0 & -1 \\ 1 & 0 \end{psmallmatrix}$ is a $\pi/2$ rotation in phase space, a.k.a.\ a Fourier transform.
Under this transformation the $\bf Z$-graph becomes $\mathbf{Z}'=\mathbf{V}'+i\mathbf{U}'$, with~\cite{Menicucci2011} 
\begin{align}\label{eq:vp}
\mathbf{V}^\prime&=
\begin{pmatrix}
	0 & 0 & t & 0 & 0 & 0 \\
	0 & 0 & t & 0 & 0 & 0 \\
	t & t & 0 & t & 0 & 0 \\
	0 & 0 & t & 0 & t & 0 \\
	0 & 0 & 0 & t & 0 & t \\
	0 & 0 & 0 & 0 & t & 0 
\end{pmatrix}\\
\mathbf{U}^\prime&=
\begin{pmatrix}
	c^{-3} & 0 & 0 & 0 & 0 & 0 \\
	0 & c^{-3} & 0 & 0 & 0 & 0 \\
	0 & 0 & c^{-1}  & 0 & 0 & 0 \\
	0 & 0 & 0 & c^{-1}  & 0 & 0 \\
	0 & 0 & 0 & 0 & c^{-1}  & 0 \\
	0 & 0 & 0 & 0 & 0 & c^{-3} 
\end{pmatrix}.
\end{align}
Now that $\bf U'$ is diagonal (still with $\tr{U'}\to0$), no hidden entanglement is present and the $\bf V'$ graph in \fig{DiscGraph}, bottom, shows the complete picture, which includes the expected 3-4 and 4-5 entanglement edges. 

As we mentioned earlier, the previous best practice for calculating the closest graph $\mathbf{V'}$ approximated by a given Gaussian state, minimizing $\tr{U}$ solely by local rotations \cite{Menicucci2011}, doesn't suffice here. As we have shown in these two examples, one must strive to make $\bf U$ diagonal by using all possible GLUs, including squeezing and shearing.

Finally, it is important to realize that this procedure doesn't always succeed: there are cases where hidden entanglement is irreducible and cannot be transferred to $\bf V$ under GLUs, i.e., where there exists no GLUs that can make $\bf U$ diagonal. We now examine these cases.

\section{Sufficient criterion for detecting irreducible hidden entanglement}\label{sec:crit}

In this section, we address the question of detecting irreducible hidden entanglement, i.e., the Gaussian states whose matrix $\bf U$ {\em cannot} be diagonalized by GLUs. Such states consequently cannot be made GLU-equivalent to a CV cluster state.

If $\mathbf{U}$ is diagonal, we can write, from \eq{covQ},
\begin{align}\label{eq:lambda}
\text{Cov}[\mathbf Q] &=  \text{diag}\{\lambda_1,\lambda_2,...,\lambda_N\},
\end{align}
where $\lambda_j>0$, $j=1,2,...,N$. This yields
\begin{align}
2\left\langle Q_jQ_k\right\rangle&=\lambda_j\delta_{jk}\\
\left\langle \{Q_j,P_k\}\right\rangle&=\lambda_j{V}_{jk}\\
2\left\langle P_jP_k\right\rangle&=\lambda_j^{-1}\delta_{jk}+\sum_{l}{V}_{jl}\lambda_l{V}_{lk}.\label{eq:last}
\end{align}
We now use the more convenient mode-ordered covariance matrix, defining $\mathbf{x}=(Q_1,P_1,Q_2,P_2,...,Q_N,P_N)^{T}$,
\begin{align}\label{eq:CMModes}
	\mathbf{\Sigma}^{(mode)}&=\text{Cov}[\mathbf{x}]
	=\begin{pmatrix}
		\mathbf{\sigma}_{11}&\mathbf{\sigma}_{12}&\cdots&\mathbf{\sigma}_{1n}\\
		\mathbf{\sigma}_{21}&\mathbf{\sigma}_{22}&\cdots&\mathbf{\sigma}_{2n}\\
		\vdots				 &\vdots			   &\ddots&\vdots\\
		\mathbf{\sigma}_{n1}&\mathbf{\sigma}_{n2}&\cdots&\mathbf{\sigma}_{nn}
	\end{pmatrix},
\end{align}
where
\begin{align}
	\mathbf{\sigma}_{jk}=&\frac{1}{2}\begin{pmatrix}
		\left\langle \{Q_j,Q_k\}\right\rangle&\left\langle \{Q_j,P_k\}\right\rangle\\
		\left\langle \{P_j,Q_k\}\right\rangle&\left\langle \{P_j,P_k\}\right\rangle
	\end{pmatrix}\\
	=&\frac{1}{2}\begin{pmatrix}
		(\mathbf{U}^{-1})_{jk}&\left(\mathbf{U}^{-1}\mathbf{V}\right)_{jk}\\
		(\mathbf{V}\mathbf{U}^{-1})_{jk}&(\mathbf{U}+\mathbf{V}\mathbf{U}^{-1}\mathbf{V})_{jk}
	\end{pmatrix},\label{eq:sigmajk}
\end{align}
which means that $\sigma_{jj}$ is the covariance matrix of qumode $j$ and $\sigma_{j\neq k}=\sigma_{k\neq j}^{T}$ contains all correlations between qumodes $j$ and $k$. In this ordering, the direct-sum GLU of \eq{glu} can be written in block-diagonal form,
\begin{align}
	\mathbf{S}_{local}=\begin{pmatrix}
		\mathbf{S}_1&\mathbb0&\cdots&\mathbb0\\
		\mathbb0&\ddots&\ddots&\vdots\\
		\vdots&\ddots&\ddots&\mathbb0\\
		\mathbb0&\cdots&\mathbb0&\mathbf{S}_N
	\end{pmatrix}
\end{align}
Under GLUs, the covariance matrix blocks evolve as 
\begin{align}
	\mathbf{\sigma}_{jk}\mapsto\mathbf{\sigma}_{jk}^\prime=\mathbf{S}_{j}\mathbf{\sigma}_{jk}\mathbf{S}_{k}^\text{T}\label{eq:sigmajktransform}.
\end{align}
This allows us to establish a first important theorem:
\begin{theorem*}
	The determinant of the correlation matrix $\sigma_{jk}$, Det[$\sigma_{jk}$] is invariant under GLUs.
\end{theorem*}
\begin{proof}
	This follows from the fact that the determinant of any symplectic matrix $\mathbf{S}$ is Det[$\mathbf{S}$]=1. Then, \eq{sigmajktransform} yields Det[$\sigma_{jk}^\prime$]=Det[$\sigma_{jk}$].
\end{proof}
The next theorem provides a sufficient, although not necessary, criterion for whether a given Gaussian state is not GLU-equivalent to a cluster state.
\begin{theorem*}
	If a given Gaussian state is GLU-equivalent to a graph state, then  Det$[\sigma_{jk}]\leqslant 0$, $\forall (j,k\neq j)$.
\end{theorem*}
\begin{proof}
	Let's go back to the $\bf U$ diagonal case. Using \eqlist{lambda}{last}, we can write
\begin{align}
	\mathbf{\sigma}_{jj}&=\begin{pmatrix}
		\lambda_j&0\\
		0&\text{Cov}[\mathbf P]_{jj}
	\end{pmatrix}\\
	\mathbf{\sigma}_{j\neq k}&=\begin{pmatrix}
		0&\lambda_j{V}_{jk}\\
		\lambda_k{V}_{jk}&\text{Cov}[\mathbf P]_{jk}
	\end{pmatrix} \label{eq:sigmajkUdiag}
\end{align}
where we have assumed, without losing any generality, that we do not have any self loops on our real graph, i.e., ${V}_{jj}=0$, $\forall j$.\footnote{This doesn't detract from the generality of our approach because any $\mathbf{V}_{jj}\neq0$ can be removed via GLU $\mathbf{S}_{j}=\begin{psmallmatrix}1&0\\-\mathbf{V}_{jj}&1\end{psmallmatrix}$.}
From \eq{sigmajkUdiag}, we deduce that the determinant of $\sigma_{jk}, j\neq k$ for $\bf U$ diagonal is  
\begin{align}\label{eq:neg}
\text{Det}[\sigma_{jk}]&=-\lambda_j\lambda_k{V}_{jk}^2\leqslant0.
\end{align}
Note that $\text{Det}[\sigma_{jk}]$=0 $\Leftrightarrow$ $V_{jk}$=0. Since Det[$\sigma_{jk}$] is GLU-invariant, any Gaussian state that is GLU-equivalent to a cluster state (i.e., has $\mathbf{U}$ diagonal) has nonpositive Det$\left[\sigma_{jk}\right], \forall j\neq k$. 
\end{proof}

The theorem is not a logical equivalence because there are more submatrices $\sigma_{jk}$ to process than available GLUs when $N(N-1)/2>N$, i.e., when the total number of qumodes $N>3$. In this case, it is possible to have states with Det$[\sigma_{jk}]\leqslant 0$, $\forall (j,k\neq j)$ that are nonetheless not GLU-equivalent to a cluster state.

Exceptions to this are the cases $N$=2, for which all two-mode pure states are GLU-equivalent to two-mode squeezed states and also to two-mode cluster states, and $N$=3, for which the number of submatrices $\sigma_{jk}$ is $N(N-1)$/2 = 3 and 3-mode Gaussian states with Det$[\sigma_{jk}]\leqslant 0$, $\forall (j,k\neq j)$ are therefore GLU-equivalent to cluster states (See supplemental material, Appendix B). 

The contraposition of the theorem above yields the sufficient criterion for Gaussian states that cannot be expressed as graph states, because $\bf U$ cannot be diagonalized by GLUs: \\

{\em If $\exists (j,k\neq j)$ such that Det$[\sigma_{jk}]>0$, then the corresponding Gaussian state is not GLU-equivalent to a graph state.}\\\ \\
To use this criterion, one just has to check all off-diagonal 2$\times$2 minors of the mode-ordered covariance matrix: a single positive Det$[\sigma_{jk}]$ is proof that irremovable hidden entanglement exists between $j$ and $k$ and that the corresponding state can never be GLU-equivalent to a cluster state.

\section{General algorithm for diagonalizing U by Gaussian local unitaries}\label{sec:diag}

We now extend the analysis of the previous section to deriving a general algorithm for diagonalizing $\bf U$ using GLUs. 

The algorithm is based on noticing, from the quadrature-ordered covariance matrix of \eqs{cova}{U-1}, that the complete absence of amplitude correlations, $\av{Q_{j}Q_{k}}=0$, $\forall$ ($j$,$k$), is logically equivalent to a diagonal $\mathbf U^{-1}$ and hence to a diagonal $\bf U$. We now introduce the GLU aspect of the diagonalization algorithm, which is best seen in the mode-ordered covariance matrix of \eqlist{CMModes}{sigmajk}: we require a set of GLUs which cancel the upper left entry of all 2$\times$2 submatrices $\sigma_{jk}$. Because there are up to $N(N-1)/2$ such submatrices and only $N$ available GLUs, it's clear that the algorithm cannot always be successful: for $N$$>$3, canceling multiple $\av{Q_{j}Q_{k}}$ with a single GLU will be required, thereby placing symmetry constraints on the Gaussian state under consideration. A detailed study of such constraints is outside the scope of this paper.

One should note that the algorithm will succeed for all Gaussian states that are GLU-equivalent to cluster states (green subset of \fig{van}) and will fail for all the rest (red subset of \fig{van}), including those undetected by the sufficient criterion derived in \sect{crit}.

We first proceed to define a ``standard'' covariance matrix of a valid Gaussian graph state (i.e.,  for which $\mathbf{U}$ is diagonal). We draw inspiration from the definition of a standard covariance matrix for LU-equivalent Gaussian states  by Adesso~\cite{Adesso2006} and by Giedke and Kraus~\cite{Giedke2014}.  We first define the single mode squeezing operations
\begin{align}
	\mathbf{T}_j=
	\begin{pmatrix}\lambda _j^{-1/2}
		 & 0 \\
		0 & \lambda _j^{1/2} 
	\end{pmatrix},\quad j=1,...,N.
\end{align}
Then, using the transformation rule of \eq{sigmajktransform}, we define the standard covariance matrix as
\begin{align}
	\tilde{\mathbf{\sigma}}_{jj}&=\mathbf{T}_j\mathbf{\sigma}_{jj}\mathbf{T}_j^{\text{T}}=\begin{pmatrix}
		1&0\\
		0&\lambda_j\text{Cov}[\mathbf P]_{jj}
	\end{pmatrix}\\
	\tilde{\mathbf{\sigma}}_{j\neq k}&=\mathbf{T}_j\mathbf{\sigma}_{jk}\mathbf{T}_k^{\text{T}}=\sqrt{\lambda_j\lambda_k}\begin{pmatrix}
		0&{V}_{jk}\\
		{V}_{jk}&\text{Cov}[\mathbf P]_{jk}
	\end{pmatrix} \label{eq:StandFormCorrM}
\end{align}
where the off-diagonal terms of $\tilde{\mathbf{\sigma}}_{jk}$ are the same and equal to $\pm\sqrt{-\text{Det}[\mathbf{\sigma}_{jk}]}$. Note that it is easy to see that this form corresponds to a diagonal $\bf U$ by virtue of \eq{sigmajk}. 

The goal is to find $N$ single-mode GLUs that take all covariance submatrices 
\begin{align}
	\sigma_{j\neq k}=
	\begin{pmatrix}
		a_{jk} & b_{jk} \\
		c_{jk} & d_{jk} \\
	\end{pmatrix}
\end{align} 
to  the standard form of \eq{StandFormCorrM}. Each of the $N$ sought GLUs is written as an Iwasawa decomposition~\cite{Arvind1995,Simon1988,Arvind1995a}
\begin{align}
	\mathbf{S}_j=
	\begin{pmatrix}
		1 & 0 \\
		q_j & 1 
	\end{pmatrix}
	\begin{pmatrix}
		r_j & 0 \\
		0 & r_j^{-1}
	\end{pmatrix}
	\begin{pmatrix}
		\cos \phi _j & -\sin \phi _j \\
		\sin \phi _j & \cos \phi _j 
	\end{pmatrix},\label{eq:SMSIwasawa}
\end{align}
where $r_j>0$. The leftmost matrix in \eq{SMSIwasawa} corresponds to a shearing transformation which shifts the diagonal elements of $\mathbf{V}$: ${V}_{jj}\mapsto{V}_{jj}+q_j$. This doesn't play any role in the process of diagonalizing $\mathbf{U}$, so without loss of generality we can set $q_j=0$, allowing us to reduce the free parameters of $\mathbf{S}_j$ to only two. The second matrix, where $r_j>0$, corresponds to a single mode squeezing and the third to a single mode rotation. In order to find GLUs that diagonalize $\mathbf{U}$, we proceed by sequentially finding the parameters $r_j$ and $\phi_j$  in \eq{SMSIwasawa} by operating one by one in all the correlation submatrices $\sigma_{j\neq k}$ and taking them to the standard form of \eq{StandFormCorrM}. It is important to point out that two GLUs $\mathbf{S}_j$ and $\mathbf{S}_k$ act simultaneously on each given $\sigma_{j\neq k}$, as shown by \eq{sigmajktransform}; hence, at each stage of the process of eliminating the term $\av{Q_{j}Q_{k}}$,  both GLUs $\mathbf{S}_j$ and $\mathbf{S}_k$ must be taken into account at the same time. This procedure varies slightly depending on whether Det$[\sigma_{j,k}]< 0$, presented below, or Det$[\sigma_{j,k}]= 0$, presented in Appendix C (Supplemental Material).

Assuming that Det$[\sigma_{j,k}]\neq 0$, let $\mathbf{M}=\begin{psmallmatrix} A & B \\C & D \\\end{psmallmatrix}$ be a general $2\times2$ matrix. If we want to use the GLU multiplying to the left in (\ref{eq:sigmajktransform}),  we find a GLU $\mathbf S_\text{left}$ such that
\begin{align}	
	\mathbf{S}_\text{left}\mathbf{M}=
	\begin{pmatrix}
		0 & \delta  \\
		\delta  & \delta\,\frac{A B+C D}{A^2+C^2} 
	\end{pmatrix},
	\label{eq:LeftDiag}
\end{align}
where $ \delta=\sqrt{-\text{Det}[\mathbf{M}]}=\sqrt{B C-A D}$. The Iwasawa parameters  are 
\begin{gather}\label{eq:LeftDiagPar}
%	\phi_\text{left}=\arccos\left(\frac{c}{\sqrt{a^2+c^2}}\right)\\
	\cos (\phi_\text{left} )= \frac{C}{\sqrt{A^2+C^2}};\quad\sin (\phi_\text{left} )= \frac{A}{\sqrt{A^2+C^2}}\\
	r_\text{left}= \delta^{-1}\,\sqrt{A^2+C^2}. 
\end{gather}
Likewise, we can use a GLU to the right, 
\begin{align}\label{eq:RightDiag}	
	\mathbf{M}\mathbf{S}_\text{right}^\text{T}=
	\begin{pmatrix}
		0 & \delta  \\
		\delta  & \delta\,\frac{A C+B D}{A^2+B^2} 
	\end{pmatrix}
\end{align}
with \begin{gather}\label{eq:RightDiagPar}
	\cos (\phi_\text{right} )= \frac{B}{\sqrt{A^2+B^2}};\quad\sin (\phi_\text{right} )= \frac{A}{\sqrt{A^2+B^2}}\\
	r_\text{right}= \delta^{-1}\,\sqrt{A^2+B^2}. 
\end{gather}

If we set $\mathbf{M}=\sigma_{jk}\mathbf{S}_k^\text{T}$ in \eqs{LeftDiag}{LeftDiagPar}, then we have taken $\sigma_{jk}$ to the standard form only using $\mathbf{S}_j$ regardless of the value that $\mathbf{S}_k$ may take. If $\mathbf{S}_k$ was already determined in a previous stage, then (\ref{eq:LeftDiagPar}) fixes the value of $\mathbf{S}_j$. In the other hand, if $\mathbf{S}_k$ have not been determined yet, then we will have that $\mathbf{S}_j$ will be a function of $\mathbf{S}_k$ (and $\sigma_{jk}$), denoted as $\mathbf{S}_j\left[\mathbf{S}_k \right]$ and $\mathbf{S}_k$ can be used to take another submatrix $\sigma_{kl}, l\neq j$ to the standard form.
Using other matrices $\sigma_{mj}, m\neq k$ and taking them to the standard form via the procedure above we can find expressions of the form $\mathbf{S}_m=\mathbf{S}_m\left[\mathbf{S}_j \right]=\mathbf{S}_m\left[\mathbf{S}_k \right]$ (where we have taken into account that $\mathbf{S}_j=\mathbf{S}_j\left[\mathbf{S}_k \right]$). Finally the value of $\mathbf{S}_k$ can be found using a correlation matrix of type  $\sigma_{km}$ where the $\alpha_{km}^\prime$ term must to be cancel out from the equation $\mathbf{\sigma}_{km}^\prime=\mathbf{S}_{k}\mathbf{\sigma}_{km}\mathbf{S}_{m}[\mathbf{S}_{k}]^\text{T}$.
Explicit forms of $\mathbf{S}_j\left[\mathbf{S}_k \right]$, more details of this procedure and how it is modified when we have singular matrices are found in Appendix C.

The algorithm is complete when all the $N$ GLUs are determined using the above procedure which, in general, makes use of a subset of all the correlation matrices $\sigma_{jk}$. As shown in Appendix C, the case may present itself where we find multiple solutions for the GLUs $\mathbf{S}_j$ $j=1,...,N$, but those solutions are the most general ones that take the correlation submatrices  $\sigma_{jk}$ used through the algorithm to the standard form. Therefore the last step is to apply those solutions to the remaining correlation matrices that were not used yet and test whether or not they all also take the standard form. If one of these GLUs sets successfully zeroes out all $\av{Q_{j}Q_{k}}$ terms, $\forall$ ($j$,$k$), then the corresponding Gaussian state is a valid (GLU-equivalent) CV cluster state, else there exists irreducible hidden entanglement.

\section{Conclusion}

We have reported and fully analyzed a fundamental effect that is specific of continuous-variable quantum information and doesn't occur in qubit-based quantum information. This effect is linked to the description of any Gaussian Wigner function, in the Iwasawa decomposition, by a complex graph state~\cite{Menicucci2011} whose stabilizers aren't necessarily unitary---or, equivalently, whose nullifiers aren't necessarily Hermitian. In order for the graph to provide a valid representation of a cluster state---with unitary stabilizers and Hermitian nullifiers---the imaginary part of the graph must vanish. Whereas all previous work did ensure that this always occurred in the infinite squeezing limit, our work shows that this limit isn't reliable as state asymmetries can give rise to hidden entanglement, at times irreducibly so. Such hidden entanglement may morph into a more versatile and useful cluster state but will disrupt quantum information processing if not accounted for. We note that an interesting extension of this work might be to explore its possible connections to an earlier study of tripartite entanglement by Adesso and Illuminati, in which pure, symmetric three-mode Gaussian states were shown to be simultaneous CV analogues of both the GHZ and the W states of three qubits~\cite{Adesso2006a}.

A simple occurrence of this effect is present in the case of the workhorse of entanglement generation: the interference of two orthogonally squeezed quantum quadratures at a balanced beamsplitter~\cite{Ou1992}, which will present excess quantum noise, as per \eq{badtms}, if the initial squeezed modes have different parameters. This is a heretofore undiscovered source of imperfections in realistic CVQI experiments. This result emphasizes the importance of forgoing the use of the infinite squeezing limit as more than a simplification when dealing with continuous-variable quantum information. An example of this philosophy is a recent result on the fault tolerance of statistical mixtures of cluster states in CVQC~\cite{Walshe2019}. 

We emphasize here that our results do not alter the feasibility of continuous-variable quantum computing with finitely squeezed Gaussian states. Indeed, it is well known that these Gaussian resources must be completed by non-Gaussian ones, necessary for exponential speedup~\cite{Mari2012,Bartlett2002} and fault tolerance~\cite{Menicucci2014ft}. Moreover, one should also remember that the resilience of cluster states under measurement~\cite{Briegel2001} makes it straightforward to cut out of a Gaussian complex graph its irreducibly imaginary edges (if they are reasonably local) so as to make the graph real and therefore a valid CV cluster state.

Our work constitutes the first concrete analytic correspondence between arbitrary Gaussian states and CV cluster states. This paves the way to deriving mappings between cluster states and the universal Bloch-Messiah decomposition of Gaussian states~\cite{Braunstein2005} which opens up a new area of research of translating nonconventional quantum circuits, such as Gaussian boson sampling~\cite{Hamilton2017}, into cluster states and one-way quantum computing.

This work is also relevant to quantum state engineering given that it allows us to identify when a given general scalable source of entangled gaussian states can be also a source of CV cluster states, which are a platform for universal quantum computation.

We thank Nicolas Menicucci, Rafael Alexander, and Israel Klich for stimulating discussions. We are grateful to one of the referees for pointing out the work of Ref.~\citenum{Adesso2006a}. This work was supported by NSF grant PHY-1820882, by the College of Arts and Sciences of the University of Virginia, by FAPESP SPRINT grant 2016/50468-7, and by the Jefferson Lab LDRD project No.~LDRD21-17 under which Jefferson Science Associates, LLC, manages and operates Jefferson Lab.

\bibliography{Pfister}
\bibliographystyle{bibstyleNCM}

\newpage
\onecolumngrid
\appendix

\section{Hidden entanglement and diagonalization of U by GLUs in the two-mode case}

An ideal two-mode cluster state is defined as two $p=0$ momentum eigenstates coupled by a $CZ=\exp(iQ_{1}Q_{2})$ gate
\vglue .001in
\begin{equation}\label{eq:C}
\setlength{\unitlength}{.35in}
\begin{picture}(8,0)
\thicklines
\put(-.5,-.1){$\ket{\mathcal C}_{12}=$}
\put(1.75,0){\circle{1.25}}
\put(1.25,-.1){$\ket{0}_{p1}$}
\put(6.75,0){\circle{1.25}}
\put(6.25,-.1){$\ket{0}_{p2}$}
\put(2.4,0){\line(1,0){3.7}}
\put(3.25,.35){${\rm C_{Z}}=e^{iQ_{1}Q_{2}}$}
\end{picture}
\end{equation}
\vglue .1in
which admit the following nullifiers 
\begin{align}\label{eq:V}
(\mathbf P-\mathbf V\mathbf Q)\ket{\mathcal C}_{12} = \mathbf0\ket{\mathcal C}_{12},
\end{align}
where $\mathbf Q$=($Q_1$,$Q_2)^{\rm T}$, $\mathbf P$=($P_1$,$P_2)^{\rm T}$, $\mathbf V=\begin{psmallmatrix} 0&1\\ 1&0\end{psmallmatrix}$.

In this appendix we study how different finitely squeezed two mode states relate with this ideal cluster state as we approach to the infinite squeezing limit.

We focus on the simple but fundamental case of the state created by the interference at a balanced beamsplitter of two quadrature eigenstates out of phase by  $\theta$.
\begin{align}\label{eq:Btheta}
 \ket{\mathcal B(\theta)}_{12} = B_{12}  \ket0_{\theta1}\ket0_{p2},
\end{align}
where $B_{12}$ = $\exp[-i\frac\pi4(a_1^\dag a_2+a_1a_2^\dag)]$ and where the generalized-quadrature nullifier and eigenstate are
\begin{align}
    A_1(\theta)\ket0_{\theta1}= (\cos\theta\, P_1+\sin\theta\, Q_1)\ket0_{\theta1}=0\ket0_{\theta1}.
\end{align}
The nullifiers of  $\ket{\mathcal B(\theta)}_{12}$ are 
\begin{align}
\mathcal N_1&=B_{12}P_2 B_{12}^\dagger =P_1-P_2\label{eq:n01}\\
\mathcal N_2(\theta)   &=B_{12}A_1(\theta) B_{12}^\dagger =\sin\theta\,(Q_1+Q_2)+\cos\theta\, (P_1+P_2).\label{eq:n02}
    \end{align}

For $\theta=0$, these nullifiers are $P_1\pm P_2$. As the stabilizers form a multiplicative group, the nullifiers form an additive one and a linear combination of nullifiers is a nullifier. Hence $P_{1,2}$ nullify $\ket{\mathcal B(0)}_{12}$, which entails $\mathbf V=\mathbb0$: 
\begin{equation}\label{eq:B0}
\setlength{\unitlength}{.35in}
\begin{picture}(8,0)
\thicklines
\put(-.5,-.1){$\ket{\mathcal B(0)}_{12}=\ket0_{p1}\ket0_{p2}=$}
\put(5.,0){\circle*{.5}}
\put(7.5,0){\circle*{.5}}
\end{picture}
\end{equation}
An edge between two vertices signifies entanglement, its absence signifies separability, and  GLUs cannot transform two separated subgraphs into a connected one.

For $\theta=\pi/2$, we recover the nullifiers of the Einstein-Podolsky-Rosen (EPR) state~\cite{Einstein1935}
\begin{align}
(Q_{1}-Q_{2})\ket{\mathcal E}_{12}&=0\ket{\mathcal E}_{12}  \label{eq:e1} \\
(P_{1}+P_{2})\ket{\mathcal E}_{12}&=0\ket{\mathcal E}_{12}, \label{eq:e2}
\end{align} 
which entails
\begin{align}\label{eq:EPRlimit}
    \ket{\mathcal B({\scriptstyle\frac\pi2})}_{12}=\ket{\mathcal E}_{12}\overset{GLU}\sim\ket{\mathcal C}_{12}.
\end{align} 

For $\theta<\pi/2$, we can rewrite the nullifiers of \eqs{n01}{n02} in cluster state form, 
\begin{align}
\mathbf V = \frac12\begin{pmatrix} \tan\theta & \tan\theta \\ \tan\theta & \tan\theta \end{pmatrix},
\end{align}
which corresponds to the graph
\begin{equation}\label{eq:B}
\setlength{\unitlength}{.35in}
\begin{picture}(8,1.)
\thicklines
\put(-.5,-.1){$\ket{\mathcal B(\theta)}_{12}=$}
\put(3.25,0){\circle*{.5}}
\qbezier(3.25,0)(2.05,.75)(2,0)
\qbezier(3.25,0)(2.05,-.75)(2,0)
\put(6.75,0){\circle*{.5}}
\qbezier(6.75,0)(7.95,.75)(8,0)
\qbezier(6.75,0)(7.95,-.75)(8,0)
\put(3.5,0){\line(1,0){3}}
\put(4.5,.35){$\scriptstyle \frac12\tan\theta$}
\put(1.75,.75){$\scriptstyle \frac12\tan\theta$}
\put(6.75,.75){$\scriptstyle \frac12\tan\theta$}
\end{picture}
\end{equation}
State $\ket{\mathcal B(\theta)}_{12}$ is therefore GLU-equivalent to a cluster state, as expected since all bipartite entangled states are equivalent under GLUs. 

We now turn to the generalization to Gaussian states of the graphical formalism introduced above, which  was developed by Menicucci, Flammia, and van Loock~\cite{Menicucci2011}. The gist of this formalism is that any pure Gaussian state can be represented by a unique graph whose edges are weighted by {\em complex} numbers, the imaginary edges representing the effect of the squeezing. This is because the nullifier of a phase-squeezed state of finite squeezing parameter $r$ isn't $P$ any more but $P-i\,e^{-2r}Q$. As a result the adjacency matrix {\bf Z} of the graph becomes complex 
\begin{equation}
{\bf Z}=\mathbf V+i\mathbf U,
\end{equation}
where $\mathbf V$ is as before. Matrix $\mathbf U$ is symmetric, like $\mathbf V$, and also positive definite. It represents the effects of finite squeezing~\cite{Menicucci2011} and can be interpreted as the error of the Gaussian state in approximating an ideal graph state of weighted adjacency matrix $\mathbf{V}$, as per
\begin{equation}\label{eq:cov}
{\rm Cov}({\bf P} - {\bf V} {\bf Q}) = \frac12{\bf U},
\end{equation}
which generalizes \eq V, Cov denoting the covariance matrix of the nullifiers. Hence,  a  Gaussian state {\bf Z} is a good approximation of a graph state $\mathbf{V}$ if $\mathbf{U}$$\to$$\mathbb{0}$, or $\tr U$$\to$0 since $\mathbf{U}$ is positive definite.

The canonical Gaussian cluster state has two  phase-squeezed qumodes (realistic implementations of phase-quadrature eigenstates, i.e. finite squeezing version of \eq C), of respective squeezing parameters $r_{1,2}$, linked by controlled-phase gates:
\begin{align}
\mathbf Z_\mathcal{C}\equiv\begin{pmatrix}0&1\\ 1&0\end{pmatrix} +i \begin{pmatrix}e^{-2r_1}&0\\ 0&e^{-2r_2}\end{pmatrix},
\end{align}
\begin{equation}\label{eq:Cr}
\setlength{\unitlength}{.35in}
\begin{picture}(8,0.75)
\thicklines
\put(-.5,-.1){$\ket{\mathcal C(r_1,r_2)}_{12}=$}
\put(3.55,0){\circle*{.5}}
\qbezier(3.55,0)(2.35,.75)(2.3,0)
\qbezier(3.55,0)(2.35,-.75)(2.3,0)
\put(6.75,0){\circle*{.5}}
\qbezier(6.75,0)(7.95,.75)(8,0)
\qbezier(6.75,0)(7.95,-.75)(8,0)
\put(3.75,0){\line(1,0){3}}
\put(5.05,.35){$\scriptstyle 1$}
\put(2.25,.7){$\scriptstyle i\,e^{-2r_1}$}
\put(7.,.7){$\scriptstyle i\,e^{-2r_2}$}
\end{picture}
\end{equation}

Our next example is the  Gaussian EPR state, a.k.a.\ the two-mode squeezed state $\ket{\mathcal E(r)}_{12}$, where $r$ is the squeezing parameter~\cite{Ou1992}. Solving the Heisenberg equations for two-mode squeezing Hamiltonian $i\frac r\tau a_1^\dag a_2^\dag +$ H.c.\ over time $\tau$, one finds~\cite{Reid1989}
\begin{align}
\Delta\left(Q_1-Q_2\right) =\Delta\left(P_1+P_2\right) = e^{-r},
\end{align}
which coincide with the nullifiers of \eqs{e1}{e2} in the infinite squeezing limit $r$$\to$$\infty$. Using Gaussian graphical calculus, we find, after a $\frac\pi2$ rotation of one mode~\cite{Menicucci2011}
\begin{align}\label{eq:T1}
\mathbf Z_\mathcal{E}\equiv\begin{pmatrix}0&\tanh2r\\ \tanh2r&0\end{pmatrix} +i \begin{pmatrix} \text{sech}\,2r & 0\\ 0 & \text{sech}\,2r \end{pmatrix},
\end{align}
\begin{equation}\label{eq:T2}
\setlength{\unitlength}{.35in}
\begin{picture}(8,0.75)
\thicklines
\put(-.5,-.1){$\mathcal{F}_1\ket{\mathcal E(r)}_{12}=$}
\put(3.25,0){\circle*{.5}}
\qbezier(3.25,0)(2.05,.75)(2,0)
\qbezier(3.25,0)(2.05,-.75)(2,0)
\put(6.75,0){\circle*{.5}}
\qbezier(6.75,0)(7.95,.75)(8,0)
\qbezier(6.75,0)(7.95,-.75)(8,0)
\put(3.5,0){\line(1,0){3}}
\put(4.5,.35){$\scriptstyle \tanh2r$}
\put(1.95,.7){$\scriptstyle i\,\text{sech}\,2r$}
\put(6.75,.7){$\scriptstyle i\,\text{sech}\,2r$}
\end{picture}
\end{equation}

Note that, in the limit $r$$\to$$\infty$ and up to a GLU (here a Fourier transform: FT), \eqs{Cr}{T2} are identical. 

The next case is much less trivial. The finitely squeezed version of $\ket{\mathcal B(\theta)}_{12}$, \eq{Btheta}, is 
\begin{align}\label{eq:Bthetar}
 \ket{\mathcal B(\theta,r_1,r_2)}_{12} =   B_{12}R_1(\theta)S_1(r_1)S_2(r_2)\ket0_{1}\ket0_{2},
\end{align}
where the initial state is vacuum, $S(r)$=$\exp[\frac r{2}(a^{\dag2}-a^2)]$ is a phase squeezing operator for $r$$>$0, and $R(\theta)$=$\exp(-i\theta a^\dag a)$ is a phase-space rotation operator. 
The Gaussian graph is 
\begin{align}\label{eq:Z}
\mathbf Z\equiv v \begin{pmatrix} 1 & 1 \\ 1 & 1 \end{pmatrix} +i \begin{pmatrix} u_+ & u_- \\ u_- & u_+ \end{pmatrix},
\end{align}
\begin{equation}
\setlength{\unitlength}{.35in}
\begin{picture}(8,.75)
\thicklines
\put(-.5,-.1){$\ket{\mathcal B(\theta,r_1,r_2)}_{12}=$}
\put(4,0){\circle*{.5}}
\qbezier(4,0)(2.8,.75)(2.75,0)
\qbezier(4,0)(2.8,-.75)(2.75,0)
\put(7.,0){\circle*{.5}}
\qbezier(7.,0)(8.2,.75)(8.25,0)
\qbezier(7.,0)(8.2,-.75)(8.25,0)
\put(4.25,0){\line(1,0){3}}
\put(5.,.35){$\scriptstyle v+iu_-$}
\put(2.5,.7){$\scriptstyle v+iu_+$}
\put(7.25,.7){$\scriptstyle v+iu_+$}
\end{picture}
\end{equation}
where
\begin{align}\label{eq:v}
v&=    -\frac{\sin2\theta\sinh2r_1}{2(e^{ 2r_1 } \cos ^2\theta+e^{-2 r_1}\sin ^2\theta)}\\\label{eq:u}
u_\pm&=\frac{e^{-2r_1}}{\cos^2\theta+e^{-4r_1}\sin^2\theta}\pm e^{-2r_2}.
\end{align}
In an initial analysis of this situation for 0$<$$\theta$$<$$\frac\pi2$, one is tempted to dismiss {\bf U} altogether as its elements $u_\pm$ clearly decrease as the squeezing factors. Turning then to {\bf V}, entanglement is clearly present since $v$$\neq$0. This result is well known~\cite{Kim2002}. 

For $\theta$=$\frac\pi2$, we know the result must be a two-mode squeezed state~\cite{Furusawa1998}. However, we have
\begin{align}\label{eq:vpi2}
v&=    0\\\label{eq:upi2}
u_\pm&=e^{2r_1}\pm e^{-2r_2}.
\end{align}
Entanglement would appear to have vanished ($\mathbf V=\mathbb0$) but here the effects of finite squeezing cannot be neglected any longer: the diverging $\mathbf U$ makes $\mathbf V$ irrelevant as an approximation of a graph state by a Gaussian state. It was proposed in Ref.~\citenum{Menicucci2011} that the closest CV graph state that can be approximated by a given Gaussian state could be found by minimizing the trace of $\mathbf{U}$ by local rotations. An extremum of $\tr U$ can always be reached using $\frac\pi2$ qumode rotations, i.e., FTs.\footnote{For any $n$-mode CV with no $p-q$ correlations, the extrema of $\tr U$ after local rotations are located at $\pi/2$ rotation angle of any subset of the original $n$-modes \cite{Menicucci2011}} This property yields
\begin{align}
\mathbf{U'}&=
\begin{pmatrix}
e^{r_1 -r_2}\text{sech}(r_1 +r_2) & 0 \\
0 & e^{r_2-r_1}\text{sech} (r_1 +r_2) \\
\end{pmatrix}
\\
\mathbf{V'}&=
\begin{pmatrix}
0 & \tanh  (r_1 +r_2) \\
\tanh  (r_1 +r_2) & 0 \\
\end{pmatrix}
\end{align}
\begin{equation}
\setlength{\unitlength}{.35in}
\begin{picture}(8,1.1)
\thicklines
\put(-.5,0){$\mathcal F\ket{\mathcal B(\frac\pi2,r_1,r_2)}_{12}=$}
\put(4.5,0.1){\circle*{.5}}
\qbezier(4.5,0.1)(3.3,.76)(3.25,0.1)
\qbezier(4.5,0.1)(3.3,-.74)(3.25,0.1)
\put(7.,0.1){\circle*{.5}}
\qbezier(7.,0.1)(8.2,.76)(8.25,0.1)
\qbezier(7.,0.1)(8.2,-.74)(8.25,0.1)
\put(4.25,0.1){\line(1,0){3}}
\put(5.,.36){$\scriptstyle \tanh(r_1+r_2)$}
\put(2.7,1.){$\begin{matrix}\scriptstyle i\, e^{r_1-r_2}\times\\ \scriptstyle  \text{sech}(r_1+r_2)\end{matrix}$}
\put(7.25,1.){$\begin{matrix}\scriptstyle i\,e^{r_2-r_1}\times\\ \scriptstyle \text{sech}(r_1+r_2)\end{matrix}$}
\end{picture}
\end{equation}
which yields \eq{T2} for $r_1$=$r_2$=$r$ and \eq{C} for $r$$\to$$\infty$.

\section{Three-mode case}

In this section we will show that for the case of a three modes, the requirement Det$\{\sigma_{jk}\}<0$ is not only necessary but also sufficient for a state to be GLU to a diagonal $\mathbf{U}$ state. In \cite{Giedke2014,Adesso2006} was developed the concept of standard forms of covariance matrices, where two Gaussian states are GLU-equivalent iff they have the same standard form of their covariace matrices.

The process to calculate the standard form is to decompose each local operation in its Bloch-Messiah decomposition $\mathbf{S}_j=\mathbf{M}_j\mathbf{\Lambda}_j\mathbf{N}_j^\text{T}$ where $\mathbf{N}_j$ and $\mathbf{M}_j$ are rotation matrices and $\mathbf{\Lambda}_j$ is a diagonal squeezing matrix. Then use $\mathbf{M}_j$ and $\mathbf{\Lambda}_j$ to symplectically diagonalize all the $\sigma_{jj}$ (that is, make them proportional to the identity). And finally using $\mathbf{N}_j$ to diagonalize as many a possible $\sigma_{jk}, j\neq k$ (after applying the first two operations) matrices using it singular value decomposition (SVD) $\sigma_{jk}=\mathbf{A}_{jk}\mathbf{D}_{jk}\mathbf{B}_{jk}$ with $\mathbf{A}_{jk}$ and $\mathbf{B}_{jk}$ orthogonal matrices and $\mathbf{D}_{jk}$ diagonal. If $\sigma_{jk}$ is proportional to an orthogonal matrix, then we can diagonalized it using only one matrix (instead of two as the SVD indicates), lets say $\mathbf{N}_j$, and use the other $\mathbf{N}_k$ in the diagonalization process of another $\sigma_{k,l}$ matrix (see \cite{Giedke2014} for details).

It has been shown also that, for a three-mode Gaussian state, the standard form is a covariance matrix with no $Q-P$ correlations, that is, it has the form
\begin{align}
\mathbf{\Sigma}^{(qp)}=\begin{pmatrix}
\mathbf{\Sigma}^q&\mathbf{0}\\
\mathbf{0}&\mathbf{\Sigma}^p
\end{pmatrix}
\end{align} 
which, when written in the mode ordering, it has the form
\begin{align}
\mathbf{\Sigma}=\begin{pmatrix}
\lambda_1\mathbb{1}&\mathbf{D}_{1,2}&\mathbf{D}_{1,3}\\
&\lambda_2\mathbb{1}&\mathbf{D}_{2,3}\\
&&\lambda_3\mathbb{1}
\end{pmatrix}\label{eqn:3ModeCM}
\end{align}
where $\mathbf{D}_{jk}$ are diagonal matrices and $\mathbb{1}$ is the $2\times2$ identity matrix. We suppose that all three $\mathbf{D}_{jk}$ matrices have negative determinant then
\begin{align}
\mathbf{D}_{jk}=\begin{pmatrix}
\alpha_{jk}&0\\
0&\beta_{jk}
\end{pmatrix}\quad\text{with}\quad\alpha_{jk}\beta_{jk}<0
\end{align}
we can transform  the covariance matrix in (\ref{eqn:3ModeCM}) into a covariance matrix with diagonal $\mathbf{U}$ (i,e $\left\langle Q_jQ_k\right\rangle =0$  $j\neq k$) in two steps. We first apply the following local squeezing operations given by the GLUs
\begin{align}
\mathbf{S}_1^{(a)}=\begin{pmatrix}
\sqrt{-\frac{\alpha_{1,2}}{\beta_{1,2}}}&0\\
0&\sqrt{-\frac{\beta_{1,2}}{\alpha_{1,2}}}
\end{pmatrix}\qquad\mathbf{S}_3^{(a)}=\begin{pmatrix}
\sqrt{-\frac{\alpha_{1,3}}{\beta_{1,3}}}&0\\
0&\sqrt{-\frac{\beta_{1,3}}{\alpha_{1,3}}}
\end{pmatrix}\mathbf{S}_1^{(a)-1}\qquad\mathbf{S}_2^{(a)}=\mathbb{1}
\end{align}
which lead us to a covariance matrix of the form
\begin{align}
\mathbf{\Sigma}=\begin{pmatrix}
\lambda_1\mathbf{S}_1^{(a)}\mathbf{S}_1^{(a)\text{T}}&\mp\sqrt{-\alpha_{1,2}\beta_{1,2}}\sigma_z&\mp\sqrt{-\alpha_{1,3}\beta_{1,3}}\sigma_z\\
&\lambda_2\mathbb{1}&\mathbf{D}_{2,3}^\prime\\
&&\lambda_3\mathbf{S}_3^{(a)}\mathbf{S}_3^{(a)\text{T}}
\end{pmatrix}
\end{align}
where the upper sign correspond to $\alpha_{jk}<0, \beta_{jk}>0$ and the lower sing to the opposite situation, $\mathbf{D}_{2,3}^\prime=$diag$\{\alpha_{2,3}^\prime,\beta_{2,3}^\prime\}=$diag$\{\alpha_{2,3}\sqrt{\frac{\beta_{1,2}\alpha_{1,3}}{\alpha_{1,2}\beta_{1,3}}},\beta_{2,3}\sqrt{\frac{\alpha_{1,2}\beta_{1,3}}{\beta_{1,2}\alpha_{1,3}}}\}$ and $\sigma_z=$diag$\{1,-1\}$ is one of the Pauli matrices. In the second step we apply the local rotations given by
\begin{align}
\mathbf{S}_1^{(b)}=\mathbf{F}\mathbf{S}_2^{(b)\text{T}}\qquad\mathbf{S}_3^{(b)}=\mathbf{F}\mathbf{S}_1^{(b)\text{T}}=\mathbf{S}_2^{(b)}
\end{align}
where $\mathbf{S}_2^{(b)}$ is a rotation to be determined later, $\mathbf{F}=\begin{psmallmatrix*}[c] 0 & 1 \\ -1 & 0 \end{psmallmatrix*}$ is a $\pi/2$ rotation (also called Fourier transform due to the way it act on quantum states) and where we use the fact that rotations in a plane commute with each other. Then given that
\begin{align}
\sigma_z\mathbf{F}^\text{T}=\mathbf{F}\sigma_z=-\begin{pmatrix}
0&1\\
1&0
\end{pmatrix}\equiv-\mathbf{J}\qquad\text{and}\qquad \sigma_z\mathbf{M}=\mathbf{M}^\text{T}\sigma_z\label{eqn:SigmaZComm}
\end{align} 
for any pure 2-dimensional rotation matrix $\mathbf{M}$, we have that
\begin{align}
\mathbf{\Sigma}=\begin{pmatrix}
\lambda_1\mathbf{S}_1\mathbf{S}_1^{\text{T}}&\pm\sqrt{-\alpha_{1,2}\beta_{1,2}}\mathbf{J}&\pm\sqrt{-\alpha_{1,3}\beta_{1,3}}\mathbf{J}\\
&\lambda_2\mathbf{S}_2\mathbf{S}_2^{\text{T}}&\mathbf{S}_2^{(b)}\mathbf{D}_{2,3}^\prime\mathbf{S}_2^{(b)\text{T}}\\
&&\lambda_3\mathbf{S}_3\mathbf{S}_3^{\text{T}}
\end{pmatrix}
\end{align}
where $\mathbf{S}_j=\mathbf{S}_j^{(b)}\mathbf{S}_j^{(a)}$ is the total local transformation.
The $\mathbf{S}_2^{(b)}$ matrix is set such that the $Q-Q$ correlation $\left\langle Q_2 Q_3\right\rangle=0$ in the matrix $\mathbf{S}_2^{(b)}\mathbf{D}_{2,3}^\prime\mathbf{S}_2^{(b)\text{T}}$. For that we write it as a generic 2-dimensional rotation
\begin{align}
\mathbf{S}_2^{(b)}=\left(
\begin{array}{cc}
\cos \theta & \sin \theta \\
-\sin \theta & \cos \theta \\
\end{array}
\right)
\end{align}
then the correlation matrix between modes 2 and 3 is

\begin{align}
\mathbf{S}_2^{(b)}\mathbf{D}_{2,3}^\prime\mathbf{S}_2^{(b)\text{T}}=\left(
\begin{array}{cc}
\alpha_{2,3}^\prime  \cos ^2\theta+\beta_{2,3}^\prime  \sin ^2\theta & (\beta_{2,3}^\prime -\alpha_{2,3}^\prime ) \sin \theta \cos \theta \\
(\beta_{2,3}^\prime -\alpha_{2,3}^\prime ) \sin \theta \cos \theta & \alpha_{2,3}^\prime  \sin ^2\theta+\beta_{2,3}^\prime  \cos ^2\theta \\
\end{array}
\right)
\end{align}.

In the equation above, is possible to eliminate the $Q-Q$ term (upper left term in the matrix) only if $\alpha_{2,3}^\prime$ and $\beta_{2,3}^\prime$ have opposite sign, that is, if Det$\{\mathbf{D}_{2,3}^\prime\}=\alpha_{2,3}^\prime\beta_{2,3}^\prime<0$. In this case we can chose $\cos \theta= \frac{\sqrt{-\beta_{2,3}}}{\sqrt{\alpha_{2,3}-\beta_{2,3}}},\sin \theta= \frac{\sqrt{\alpha_{2,3}}}{\sqrt{\alpha_{2,3}-\beta_{2,3}}}$ for $\beta_{2,3}<0$ or $\cos \theta= \frac{\sqrt{\beta_{2,3}}}{\sqrt{\beta_{2,3}-\alpha_{2,3}}},\sin \theta= \frac{\sqrt{-\alpha_{2,3}}}{\sqrt{\beta_{2,3}-\alpha_{2,3}}}$ for $\alpha_{2,3}<0$ and then

\begin{align}
\mathbf{S}_2^{(b)}\mathbf{D}_{2,3}^\prime\mathbf{S}_2^{(b)\text{T}}=\left(
\begin{array}{cc}
0 & \pm\sqrt{-\alpha _{2,3} \beta _{2,3}} \\
\pm\sqrt{-\alpha _{2,3} \beta _{2,3}} & \alpha _{2,3}+\beta _{2,3} \\
\end{array}
\right)
\end{align}
where the minus sign is correspond to $\alpha_{2,3}<0$ and the plus sing to $\beta_{2,3}<0$. I this way we have gotten rid of all of diagonal terms of $\mathbf{U}$. 

In this procedure we first use single mode squeezing operations to turn the correlation matrices $\sigma_{1,2}$ and $\sigma_{1,3}$ proportional to the Pauli matrix $\sigma_z$ and the using rotations to eliminate the $Q-Q$ correlations (the first step is useful because using (\ref{eqn:SigmaZComm}) we can transform a given correlation matrix $\sigma_{jk}, j\neq k$ to a matrix proportional to $\mathbf{J}$ by using, let say, $\mathbf{S}_j$ regardless of the value of $\mathbf{S}_k$).

This procedure is not enough for a general $N$ mode Gaussian states given that the number of correlation matrices $\sigma_{jk}, j\neq k$ are $N(N-1)/2$ and we only have $N$ GLUs to eliminate all the $Q-Q$ correlations, but for $N=3$ we have $N(N-1)/2=N$.

For $N>3$ only in the case that we have enough symmetries in the state such that one local operation can in fact eliminate simultaneously several $Q-Q$ correlations and the $n$ GLUs would be enough to deal with the $N(N-1)/2$ of diagonal correlation matrices.

\section{Detailed $\mathbf{U}$ Diagonalization Procedure}

\subsection{Negative-determinant correlation submatrices}
If we start by operating in the matrix $\sigma_{12}$ then the explicit form of $\mathbf{S}_1\left[\mathbf{S}_2 \right]$ is
\begin{align}
		\mathbf{S}_1\left[\mathbf{S}_2 \right]=\begin{pmatrix}
			\displaystyle\frac{r_2}{\delta _{12}} \left(c_{12} \cos\phi _2-d_{12} \sin \phi _2\right) & \displaystyle\frac{r_2}{\delta _{12}} \left(b_{12} \sin\phi _2-a_{12} \cos\phi _2\right) \\
			\displaystyle\frac{2 \delta _{12}}{r_2} \frac{a_{12} \cos\phi _2-b_{12} \sin\phi _2}{\epsilon_{12} \sin2\phi _2+\tau_{12} \cos2\phi _2+\rho _{12}} & 
			\displaystyle\frac{2 \delta _{12}}{r_2} \frac{c_{12} \cos\phi _2-d_{12} \sin\phi _2}{\epsilon_{12} \sin2\phi _2+\tau_{12} \cos2\phi _2+\rho _{12}} \\
		\end{pmatrix}
\end{align}
where $\delta_{jk}=\sqrt{-\text{Det}[\sigma_{jk}]}=\sqrt{b_{jk} c_{jk}-a_{jk} d_{jk}}$ and
\begin{align}
	\epsilon_{12}=-2 a_{12} b_{12}-2 c_{12} d_{12}\\
	\tau_{12}=a_{12}^2-b_{12}^2+c_{12}^2-d_{12}^2\\
	\rho_{12}=a_{12}^2+b_{12}^2+c_{12}^2+d_{12}^2.
\end{align}
Then we can move throughout the first block row of the $\mathbf\Sigma^{(mode)}$, \eq{CMModes}, and bring the correlations matrices $\sigma_{1j}$ $j=3,4,...,N$ to the desired form using the GLUs $\mathbf{S}_j$ $j=3,4,...,N$ from the right as in (\ref{eq:RightDiag}) with $\mathbf{M}=\mathbf{S}_1\sigma_{1j}$. Similar to the case above, and given that $\mathbf{S}_1$ is now a function of $\mathbf{S}_2$, we end up with matrices $\mathbf{S}_j$ as functions of  $\mathbf{S}_2$, $\mathbf{S}_j=\mathbf{S}_j\left[\mathbf{S}_2\right]$. Again, all the matrices $\sigma_{1j}$ will be in the desired form regardless of the value of $\mathbf{S}_2$. Explicitly, the matrices $\mathbf{S}_j[\mathbf{S}_2]$ are 
\begin{align}
		\mathbf{S}_j\left[\mathbf{S}_2\right]=
		\begin{pmatrix}
			\displaystyle\frac{r_2}{\delta _{12} \delta _{1j}} \left(\gamma _{1j}\sin \phi _2 +\eta _{1j}\cos \phi _2 \right)
			& \displaystyle \frac{r_2}{\delta _{12} \delta _{1j}} \left(\mu _{1j}\sin\phi _2 + \nu _{1j}\cos\phi _2\right) \\
			\displaystyle-\frac{\delta _{12} \delta _{1j}}{r_2 }\, \frac{\mu _{1j}\sin\phi _2 +\nu _{1j}\cos\phi _2}{\zeta _{1j}\sin2 \phi _2 +\kappa _{1j}+\xi _{1j}\cos2 \phi _2} 
			& \displaystyle\frac{\delta _{12} \delta _{1j}}{r_2} \frac{\gamma _{1j}\sin\phi _2 +\eta _{1j}\cos\phi _2 }{\zeta _{1j}\sin2 \phi _2 +\kappa _{1j}+\xi _{1j}\cos2 \phi _2 } 
		\end{pmatrix},
		\qquad \forall j=3,4,...,N,
	\end{align}
	with
	\begin{align}
		\kappa_{1j}=&\frac{1}{2} \left[\left(c_{12}^2+d_{12}^2\right) \left(a_{1j}^2+b_{1j}^2\right)-2 \left(a_{12} c_{12}+b_{12} d_{12}\right) \left(a_{1j} c_{1j}+b_{1j} d_{1j}\right)+\left(a_{12}^2+b_{12}^2\right) \left(c_{1j}^2+d_{1j}^2\right)\right]\\
		\zeta_{1j}=&-c_{12} d_{12} \left(a_{1j}^2+b_{1j}^2\right)+\left(a_{12} d_{12}+b_{12} c_{12}\right) \left(a_{1j} c_{1j}+b_{1j} d_{1j}\right)-a_{12} b_{12} \left(c_{1j}^2+d_{1j}^2\right)\\
		\xi _{1j}=&\frac{1}{2} \left[\left(c_{12}^2-d_{12}^2\right) \left(a_{1j}^2+b_{1j}^2\right)+2 \left(b_{12} d_{12}-a_{12} c_{12}\right) \left(a_{1j} c_{1j}+b_{1j} d_{1j}\right)+\left(a_{12}^2-b_{12}^2\right) \left(c_{1j}^2+d_{1j}^2\right)\right]
\end{align}
\begin{align}
	\gamma _{1j} =b_{12} d_{1j}-d_{12} b_{1j}\\
	\eta _{1j} =c_{12} b_{1j}-a_{12} d_{1j}\\
	\mu _{1j} =d_{12} a_{1j}-b_{12} c_{1j}\\
	\nu _{1j}=a_{12} c_{1j}-c_{12} a_{1j}.
\end{align}
Once this procedure is complete we have all elements $\mathbf{U}_{1j}=0$, $j=2,3,...,N$ so now all the edges connecting qumode 1 with the rest of the graph can only be real. So far, we have set $N-1$ of the GLUs, for the last one we just take any other of the remaining correlation matrices to the desired form, \eq{sigmajkUdiag}, let's say the matrix $\sigma_{23}$,
\begin{align}
	\sigma_{23}^\prime=\mathbf{S}_2\sigma_{23}\mathbf{S}_3\left[\mathbf{S}_2\right]^\text{T}
\end{align}
It's of course enough to set the upper left element equal to zero. These have the form
\begin{align}
	a_{23}^\prime=r_2\left(A \cos2\phi_2+B\sin2\phi_2+C\right)
\end{align}
and $a_{2,3}^\prime=0$ has real solutions (two) for $\phi_2$ iff $C^2\leq A^2+B^2$. Again, which solution is the correct one will be determined by testing the solutions on the remaining correlation submatrices. 

\subsection{Singular correlation submatrices}

If $\sigma_{j,k}=0$, this does not give us any information about the required GLUs $\mathbf{S}_j$ and $\mathbf{S}_k$ so we just move to the next nonzero submatrix. If  $\sigma_{j,k}\neq0$ then, defining $\mathbf{M}=\begin{psmallmatrix} a & b \\c & d \\\end{psmallmatrix}=\sigma_{jk}\mathbf{S}_{right}^{\text{T}}$ a general $2\times2$ matrix with null determinant. Then we can operate with a sympletic matrix to the left such that the upper-right component of the correlation submatrix will set to zero (Det$[\sigma_{j,k}]=0 $ also means that the off-diagonal terms of $\sigma_{j\neq k}$ are also zero), i.e.
\begin{align}	
	\mathbf{S}_\text{left}\mathbf{M}=
	\begin{pmatrix}
		0 & 0 \\
		\frac{a b+c d}{r \sqrt{b^2+d^2}} & \frac{\sqrt{b^2+d^2}}{r} \\
	\end{pmatrix},
	\label{eq:LeftDiagDet0}
\end{align}
where the upper-left component is automatically zero too due to the fact that $ \text{Det}[\mathbf{M}]=ad-b c=0$. To achieve this we set the rotation 
\begin{align}\label{eq:LeftDiagParDet0}
\cos (\phi_\text{left} )&= \frac{d}{\sqrt{b^2+d^2}}\\ 
\sin (\phi_\text{left} )&= \frac{b}{\sqrt{b^2+d^2}}
\end{align} and the squeezing parameter $r_{left}$ would remain to be determined later in the algorithm. If $\mathbf{S}_{right}$ was already defined in a previous step of the algorithm then the lower-left term of (\ref{eq:LeftDiagDet0}) must be also zero, otherwise that means that the state is not GLU-equivalent to a cluster state. In the case that $\mathbf{S}_{right}$ was not determined yet, we can use it to impose that $ab+cd=0$. This is easily done by setting $\phi_{right}$ such that
\begin{gather}\label{eq:RightDiagParDet0}
	\cos (2 \phi_{right} )= \frac{a_{jk}^2-b_{jk}^2+c_{jk}^2-d_{jk}^2}{N_{jk}},\\
	\sin (2 \phi_{right} )= \frac{-2 \left(a_{jk} b_{jk}+c_{jk} d_{jk}\right)}{N_{jk}}\\
	N_{jk}=\sqrt{4 \left(a_{jk} b_{jk}+c_{jk} d_{jk}\right)^2+\left(a_{jk}^2-b_{jk}^2+c_{jk}^2-d_{jk}^2\right)^2}
\end{gather} where as before, the squeezing parameter $r_{right}$ can be used in another stage of the algorithm as before. It is worth pointing out that, in this case, where $ \text{Det}[\sigma_{jk}]=0$, we need both $\phi_{left}$ and $\phi_{left}$ to take $\sigma_{jk}$ to the standard form, contrary to the case  Det$[\sigma_{j,k}]\neq 0 $ as we will see below. There are two different solutions to (\ref{eq:RightDiagParDet0}) (if a given $\phi$ is a solution, then $\phi+\pi$ is also a solution corresponding to a different rotation) but which one have to be picked will be determined by applying both solutions in following steps on the algorithm.

\end{document}